%% file: main.tex
\documentclass{article}

\usepackage[final]{neurips_2021}

\usepackage[utf8]{inputenc} 
\usepackage[T1]{fontenc}    
\usepackage{hyperref}       
\usepackage{url}            
\usepackage{booktabs}       
\usepackage{amsfonts}       
\usepackage{nicefrac}       
\usepackage{microtype}      
\usepackage{xcolor}         

\hypersetup{
	colorlinks=true,
	citecolor=teal,
}

\usepackage{threeparttable}
\usepackage{makecell}

\usepackage{subcaption}
\usepackage{enumitem}

\usepackage{algorithm}
\usepackage{algorithmic}

\usepackage{amsmath}
\usepackage{amssymb} 
\usepackage{amsthm}

\usepackage[group-separator={,}]{siunitx}

\def\tcb#1{{{\color{blue}#1}}}

\def\tcr#1{{{\color{red}#1}}}

\usepackage{pgfplots}
\pgfplotsset{compat=1.9}
\input{mymacros}

\title{Communication-efficient SGD: From Local SGD to One-Shot Averaging}

\author{%
  Artin Spiridonoff
\\
  Division of Systems Engineering\\
  Boston University\\
  Boston, MA 02215\\
  \texttt{artin@bu.edu} \\
   \And
   Alex Olshevsky\\
   Division of Systems Engineering\\
   Boston University\\
   Boston, MA 02215\\
   \texttt{alexols@bu.edu} \\
   \AND
   Ioannis Ch.~Paschalidis \\
   Division of Systems Engineering\\
   Boston University\\
   Boston, MA 02215\\
   \texttt{yannisp@bu.edu} \\
}

\begin{document}

\maketitle

\begin{abstract}
	We consider speeding up stochastic gradient descent (SGD) by parallelizing it across multiple workers. We assume the same data set is shared among $N$ workers, who can take SGD steps and coordinate with a central server. While it is possible to obtain a linear reduction in the variance by averaging all the stochastic gradients at every step, this requires a lot of communication between the workers and the server, which can dramatically reduce the gains from parallelism.
	The Local SGD method, proposed and analyzed in the earlier literature, suggests machines should make many local steps between such communications. While the initial analysis of Local SGD showed it needs $\Omega ( \sqrt{T} )$ communications for $T$ local gradient steps in order for the error to scale proportionately to $1/(NT)$, this has been successively improved in a string of papers, with the state of the art requiring  $\Omega \left( N \left( \mbox{ poly} (\log T) \right) \right)$ communications. In this paper, we suggest a Local SGD scheme that communicates less overall by communicating less frequently as the number of iterations grows.  Our analysis shows that this can achieve an error that scales as $1/(NT)$ with a number of communications that is completely independent of $T$. In particular, we show that $\Omega(N)$ communications are sufficient. Empirical evidence suggests this bound is close to tight as we further show that $\sqrt{N}$ or $N^{3/4}$ communications fail to achieve linear speed-up in simulations. Moreover, we show that under mild assumptions, the main of which is twice differentiability on any neighborhood of the optimal solution, one-shot averaging which only uses a single round of communication can also achieve the optimal convergence rate asymptotically.
\end{abstract}

%
%
%
%
%

%

\section{Introduction}
Stochastic Gradient Descent (SGD) is a widely used algorithm to minimize convex functions $f$ in which model parameters are updated iteratively as 
\begin{align*}
\bx^{t+1} = \bx^t - \eta_t \hbg^t,
\end{align*}
where $\hbg^t$ is a stochastic gradient of $f$ at the point $\bx^t$ and $\eta_t$ is the learning rate.
This algorithm can be naively parallelized by adding more workers independently  to compute a gradient and then average them at each step to reduce the variance in estimation of the true gradient $\nabla f(\bx^t)$ \citep{dekel2012optimal}. This method requires each worker to share their computed gradients with each other at every iteration. We will refer to this method as "synchronized parallel SGD."

However, it is widely acknowledged that communication is a major bottleneck of this method for large scale optimization applications \citep{mcmahan2016communication,konevcny2016federated,lin2017deep}.
Often, mini-batch parallel SGD is suggested to address this issue by increasing the computation to communication ratio. Nonetheless, too large mini-batch size might degrade performance \citep{lin2018don}. Along the same lines of increasing the computation over communication effort, \textit{local} SGD has been proposed to reduce communications \citep{mcmahan2016communication,dieuleveut2019communication}. In  this method, workers compute (stochastic) gradients and update their parameters locally, and communicate only once in a while to obtain the average of their parameters.
Local SGD improves the communication efficiency not only by reducing the number of communication rounds, but also alleviates the synchronization delay caused by waiting for slow workers and evens out the variations in workers' computing time \citep{wang2018cooperative}.

On the other hand, since individual gradients of each worker are calculated at different points, this method introduces residual error as opposed to fully synchronized SGD. Therefore, there is a trade-off between having fewer communication rounds and introducing additional errors to the gradient estimates.

The idea of making local updates is not new and has been used in practice for a while  \citep{mangasarian1995parallel,konevcny2016federated}. However, until recently, there have been few successful efforts to analyze Local SGD theoretically and therefore it is not fully understood yet.
\citet{zhang2016parallel} show that for quadratic functions, when the variance of the noise is higher far from the optimum, frequent averaging leads to faster convergence.
The first question we try to answer in this work is: how many communication rounds are needed for Local SGD to have the \textit{similar} convergence rate of a synchronized parallel SGD while achieving performance that linearly improves in the number of workers?

\citet{stich2018local} was among the first who sought to answer this question for general strongly convex and smooth functions and showed that the communication rounds can be reduced up to a factor of $H = \O(\sqrt{T/N})$, without affecting the asymptotic convergence rate (up to constant factors), where $T$ is the total number of iterations and $N$ is number of parallel workers. 

Focusing on smooth and possibly non-convex functions which satisfy  a Polyak-Lojasiewicz condition, \citet{haddadpour2019local} demonstrate that only $R = \Omega((TN)^{1/3})$ communication rounds are sufficient to achieve asymptotic performance that scales proportionately to $1/N$.

Recently, \citet{khaled2019tighter} and \citet{stich2019error} improve upon the previous works by showing linear speed-up for Local SGD with only $\Omega \left(N \mbox{ poly log }(T) \right)$ communication rounds when data is identically distributed among workers and $f$ is strongly convex. Their works also consider the cases when $f$ is not necessarily strongly convex as well as the case of data being heterogeneously distributed among workers. 

More recently, \cite{FedAC} proposed a new \emph{accelerated} method that requires only $\Omega \left(N^{1/3} \mbox{ poly log }(T) \right)$ communication rounds for linear speed-up. While their results improve upon the earlier work, the communication requirements remain dependent on the total iterations $T$.

One-Shot Averaging (OSA), a method that takes an extreme approach to reducing communication, involves workers performing local updates until the very end when they average their parameters \citep{mcdonald2009efficient,zinkevich2010parallelized,zhang2013communication,rosenblatt2016optimality,godichon2020rates}. This method can be seen as an extreme case of Local SGD with $R=1$ and $H=T$ local steps. \citet{dieuleveut2019communication,godichon2020rates} provide an analysis of OSA and show that asymptotically, linear speed-up in the number of workers is achieved for a weighted average of iterates. However, both of these works make  restrictive assumptions such as uniformly three-times continuously differentiability and bounded second and third derivatives or twice differentiability almost everywhere with bounded Hessian, respectively.
The second question we attempt to answer in this work, is whether these assumptions can be relaxed and OSA can achieve linear speed-up in more general scenarios.

In this work, we focus on smooth and strongly convex functions with a  general noise model. Our contributions are three-fold:
\begin{enumerate}
	\item We propose a communication strategy which requires only $R = \Omega(N)$ communication rounds to achieve performance that scales as $1/N$ in the number of workers. To the best of the authors' knowledge, this is the only work to show that the number of communications can be taken to be completely independent of $T$. All previous papers required a number of communications which was at least $N$ times a polynomial in $\log(T)$, or had a stronger scaling with $T$. A comparison of our result  to the available literature can be found in Table \ref{table: comparison}.
	\item We show under mild additional assumptions, in particular twice differentiability on a neighborhood of  the optimal point, OSA reaches linear speed-up asymptotically, i.e., with only one communication round we achieve the convergence rate of $\O(1/(NT))$.
	\item We simulate a simple  example which is not twice differentiable at the optimum and observe that our bounds for part 1. are reasonably close  to being tight. In particular, using $1$ or $\sqrt{N}$ or $N^{3/4}$ communications does not appear to result in a linear speed-up in the number of workers (while $N$ communications does give a linear speed-up).
\end{enumerate}

\begin{table}
  \caption{Comparison of Similar Works}
  \label{table: comparison}
  \begin{threeparttable}
  	\centering
  	\resizebox{\textwidth}{!}{
  		\begin{tabular}{@{}llll@{}}
  			\toprule
  			Reference & \makecell[l]{Convergence rate \\ $f(\hat \bx^T) - f^*$\tnote{\tcb{a}}} & \makecell[l]{Communication \\ Rounds $R$}  & \makecell[l]{Noise \\ model} \\ 
  			\midrule
  			\cite{stich2018local} & $\O \left(
  			\frac{\xi^0}{R^3} + \frac{\sigma^2}{\mu N T} + 
  			\frac{\kappa G^2}{\mu R^2} \right)$\tnote{\tcb{b}}  & $\Omega(\sqrt{TN})$ & \small{uniform}  \\ 
  			\cite{haddadpour2019local} &
  			$\O \left(
  			\frac{\xi^0}{R^3}+\frac{\kappa \sigma^2}{\mu N T} + 
  			\frac{\kappa^2 \sigma^2}{\mu N T R} \right)$ & $\Omega((TN)^{1/3})$ & \makecell[l]{\small uniform with\\ \small strong-growth\tnote{\tcb{c}}} \\ 
  			\cite{stich2019error} &
  			$\tilde \O\left(
  		    \text{exp. decay} + 
  			\frac{\sigma^2}{\mu N T}\right)$\tnote{\tcb{d}} & $ \Omega(N*\text{poly} (\log T))$ & \makecell[l]{\small uniform with\\ \small strong-growth} \\ 
  			\cite{woodworth2020local} & $\O \left( \text{exp. decay} + \frac{\sigma^2}{\mu N T} + \frac{\kappa \sigma^2 \log(9+ T/\kappa)}{\mu T R}\right)$ & $ \Omega(N*\text{poly} (\log T))$ & \small{uniform}\\
  			\cite{khaled2019tighter} & $\tilde \O\left( 
  			\frac{\kappa \xi^0}{T^2} + 
  			\frac{\kappa \sigma^2}{\mu N T} + \frac{\kappa^2 \sigma^2}{\mu T R} \right)$ & $ \Omega(N*\text{poly} (\log T))$ & \small{uniform} \\ 
  			\cite{FedAC} & $\tilde \O \left(\text{exp. decay} + \frac{\sigma^2}{\mu N T} + \frac{\kappa^2 \sigma^2}{\mu T R^3} \right)$\tnote{\tcb{e}} & $ \Omega(N^{1/3}*\text{poly} (\log T))$ & \small{uniform}
  			\\
  			\textbf{This Paper} & $\O \left( 
  			\frac{(1 + c\kappa^2\ln(TR^{-2}))\xi^0}{\kappa^{-2}T^2} + 
  			\frac{\kappa \sigma^2}{\mu N T} + \frac{\kappa^2 \sigma^2}{\mu T R}\right)$\tnote{\tcb{f}} & $\Omega(N)$ & \makecell[l]{\small uniform with\\ \small strong-growth} \\ 
  			\bottomrule
  		\end{tabular}
  	}
  \begin{tablenotes}
	\item [\tcb{a}] Depending on the work, $\hat \bx^T$ is either the last iterate or a weighted average of iterates up to $T$.
  	\item [\tcb{b}] $G$ is the uniform upper bound assumed for the $l_2$ norm of gradients in the corresponding work.
  	\item [\tcb{c}] This noise model is defined in Assumption~\ref{asm: noise strong growth}.
  	\item [\tcb{d}] $\tilde \O (.)$ ignores the poly-logarithmic and constant factors.
  	\item [\tcb{e}] This is the bound for FedAC-II. FedAC-I requires $R= \Omega (N^{1/2}*\text{poly}(\log T))$.
  	\item [\tcb{f}] $c$ is the multiplicative factor in the noise model defined in Assumption~\ref{asm: noise strong growth}.
  \end{tablenotes}
  \end{threeparttable}
\end{table}

We notice that FedAC \citep{FedAC} has a better dependence on the number of workers $N$, in expense of (poly logarithmic) dependence on $T$.  
With that in mind, we still believe our communication strategy is of independent interest, particularly in the framework of non-accelerated methods. We have performed extensive numerical experiments and comparisons between the two methods and highlighted the regimes where each method outperforms the other.

It is worth mentioning that although the the communication complexity by \citet{woodworth2020local} depends on $T$, their bound has a lower dependence on condition number $\kappa$. Hence, their results are stronger than ours only when $\kappa = \Omega(\log T)$.

The rest of this paper is organized as follows. In the following subsection we outline the related literature and ongoing works. In Section  \ref{sec: Problem} we define the main problem and state our assumptions. We present our theoretical findings in Section \ref{sec: convergence} followed by numerical experiments in Section \ref{sec: Numerical } and conclusion remarks in Section \ref{sec: conclusion}.

\subsection{Related work}
There has been a lot of effort in the recent research to take into account the communication delays and training time in designing faster algorithms \citep{mcdonald2010distributed,zhang2015deep,bijral2016data,kairouz2019advances}.
See \citep{tang2020communication} for a comprehensive survey of communication efficient distributed training algorithms considering both system-level and algorithm-level optimizations.

Many works study the communication complexity of distributed methods for convex optimization \citep{arjevani2015communication,woodworth2020local}  and statistical estimation \citep{zhang2013information}.
\citet{woodworth2020local} present a rigorous comparison of Local SGD with $H$ local steps and mini-batch SGD with  $H$ times larger mini-batch size and the same number of communication rounds (we will refer to such a method as large mini-batch SGD) and show regimes in which each algorithm performs better: they show that Local SGD is strictly better than large mini-batch SGD when the functions are quadratic. Moreover, they prove a lower bound on the worst case of Local SGD that is higher than the worst-case error of large mini-batch SGD in a certain regime.
\citet{zhang2013information} study the minimum amount of communication required to achieve centralized minimax-optimal rates by establishing lower bounds on minimax risks for distributed statistical estimation under a communication budget.

A parallel line of work studies the convergence of Local SGD with non-convex functions \cite{zhou2017convergence}.
\cite{yu2019parallel} was among the first works to present provable guarantees of Local SGD with linear speed-up.
\citet{wang2018cooperative} and \citet{koloskova2020unified} present unified frameworks for analyzing decentralized SGD with local updates, elastic averaging or changing topology.
The follow-up work of \citet{wang2018adaptive} presents ADACOMM, an adaptive communication strategy that starts with infrequent averaging and then increases the communication frequency in order to achieve a low error floor. They analyze the error-runtime trade-off of Local SGD 
with nonconvex functions and propose communication times to achieve faster runtime.

Another line of work reduces the communication by compressing the gradients and hence limiting the number of bits transmitted in every message between workers \citep{lin2017deep,alistarh2017qsgd,wangni2018gradient,stich2018sparsified,stich2019error}.

Asynchronous methods have been studied widely due to their advantages over synchronized methods which suffer from synchronization delays due to the slower workers \citep{olshevsky2018robust}.
\citet{wang2019matcha} study the error-runtime trade-off in decentralized optimization and proposes MATCHA, an algorithm which parallelizes inter-node communication by decomposing the topology into matchings. 
However, these methods are relatively more involved and they often require full knowledge of the network, solving a semi-definite program and/or calculating communication probabilities (schedules) as in \citet{hendrikx2019accelerated}.

\paragraph{The homogeneous data assumption.}
In this work, we focus on the case when the data distribution is the same across workers.  A number of previous works \citep{khaled2019tighter,haddadpour2019local,stich2018local,dieuleveut2019communication} studied local SGD under this assumption. The assumption is valid when the same data set is either shared across multiple workers in the same cluster, or the assignment of data points to workers is random so that any distributional differences are small. Sharing the  data set across multiple workers in this way is a popular strategy to speed up training. For example, such data sharing is implemented in  \citep{chen2012pipelined,yadan2013multi,zhang2013asynchronous} to speed up training of deep neural networks with multiple GPUs within a single server.
While there are many widely used mechanisms such as Horovod \citep{sergeev2018horovod} for \emph{synchronous} data-parallel distributed training, they share a major communication bottleneck of broadcasting gradients to all workers \citep{grubic2018synchronous}.
Local SGD improves on these methods by reducing the communication of model parameters from every iteration to a smaller number of rounds during the entire optimization process. 
Our approach further reduces the communication overhead by communicating less as the number of iterations grows. 


\subsection{Notation}
For a positive integer $s$, we define $[s]:=\{1,\ldots,s\}$. We use bold letters to represent vectors. We denote vectors of all $0$s and $1$s by $\mathbf{0}$ and $\mathbf{1}$, respectively. We use $\Vert \cdot \Vert$ for the Euclidean norm of a vector and spectral norm of a matrix. Finally, $\mathcal{N}(\mu, \sigma^2)$ denotes a normal distribution with mean $\mu$ and variance $\sigma^2$.

\section{Problem formulation}\label{sec: Problem}
Suppose there are $N$ workers  $\V=\{1, \ldots, N\}$, trying to minimize $f:\R^d \rightarrow \R$ in parallel.
We assume all workers have access to $f$ through noisy gradients. In Local SGD, workers perform local gradient steps and occasionally calculate the average of all workers' iterates. 
Each worker $i$ holds a local parameter $\bx_i^t$ at iteration $t$. There is a set $\I \subset [T]$ of communication times and nodes perform the following update:
\begin{align}\label{eq: Local SGD process}
\bx_i^{t+1} = \begin{cases}
x_i^t - \eta_t \hbg_i^t, \qquad &\text{if } t+1 \notin \I,\\
\frac{1}{N}\sum_{j=1}^N (\bx_j^t - \eta_t \hbg_j^t),
\qquad &\text{if } t+1 \in \I,
\end{cases}
\end{align}
where $\hbg_i^t$ is an unbiased stochastic gradient of $f$ at $\bx_i^t$.
When $\I = [T]$, we recover  fully synchronized parallel SGD while $\I = \{T\}$ recovers one-shot averaging.
Pseudo-code for Local SGD is provided as Algorithm \ref{alg: Local SGD}.
\begin{algorithm}
	\caption{Local SGD}
	\begin{algorithmic}[1]\label{alg: Local SGD}
		\STATE Input: $\bx_i^0 = \bx^0$ for all $i \in [n]$, total number of iterations $T$, the step-size sequence $\{\eta_t\}_{t=0}^{T-1}$, and $\I \subseteq [T]$
		\FOR{$t=0,\ldots,T-1$}
		\FOR{$j=1,\ldots,N$}
		\STATE evaluate a stochastic gradient $\hbg_j^t$
		\IF{$t+1 \in \mathcal{I}$}
		\STATE $\bx_j^{t+1} = \frac{1}{N}\sum_{i=1}^N (\bx_i^t - \eta_t \hbg_i^t)$
		\ELSE
		\STATE $\bx_j^{t+1} = \bx_j^t - \eta_t \hbg_j^t$
		\ENDIF
		\ENDFOR
		\ENDFOR
	\end{algorithmic}
\end{algorithm}

Next we state the assumptions that we will use in our results. Note that we will not require all of them to hold at once.
\begin{assumption}[smoothness] \label{asm: smoothness}
	The function $f:\R^d\rightarrow \R$ is continuously differentiable and its gradients are $L$-Lipschitz, i.e.,
	\begin{align*}
		\Vert \nabla f(\bx) - \nabla f(\by) \Vert \leq L \Vert \bx - \by \Vert, \qquad \forall \bx, \by.
	\end{align*}
\end{assumption}

\begin{assumption}[strong convexity] \label{asm: strong convexity}
	$f$ is $\mu$-strongly convex with $\mu>0$, i.e.,
	\begin{align*}
	f(\bx) + \langle \bg, \by - \bx \rangle + \frac{\mu}{2}\Vert \bx - \by \Vert^2 \leq f(\by), \qquad \forall \bx,\by \in \R^d, \forall \bg \in \partial f(\bx),
	\end{align*}
where $\partial f(\bx)$ denotes the set of subgradients of $f$ at $\bx$. When $f$ is also continuously differentiable, $\partial f(\bx) = \{ \nabla f(\bx) \}$. 
\end{assumption}

Note that when $f$ satisfies Assumption \ref{asm: strong convexity}, it has a \emph{unique} optimal point $\bx^*$ where $f(\bx^*) = f^*$ where $f^*= \min_\bx f(\bx)$.

\begin{assumption}[Polyak-{\L}ohasiewicz condition] \label{asm: PL}
	$f$ is $\mu$-Polyak-{\L}ohasiewicz ($\mu$-PL for short) if
	\begin{align*}
	\Vert \nabla f(\bx) \Vert^2 \geq 2\mu(f(\bx) - f^*), \qquad \forall \bx.
	\end{align*}
	where $f^*= \min_\bx f(\bx)$ is the global minimum of $f$. We further assume that $f$ has a \emph{unique} optimal point $\bx^*$ where $f(\bx^*) = f^*$.
\end{assumption}
When $f$ satisfies both Assumptions \ref{asm: smoothness} and \ref{asm: strong convexity} or Assumptions \ref{asm: smoothness} and \ref{asm: PL}, we define $\kappa=L/\mu$ as the condition number of $f$.

Strong convexity implies the PL condition but the reverse does not always hold. For instance, the logistic regression loss function satisfies the PL condition over any compact set  \citep{karimi2016linear}. In fact, a PL function is not even necessarily convex. 
\cite{charles2018stability} show that deep networks with linear activation functions are PL
almost everywhere in the parameter space. \citet{allen2018convergence} show, with high probability over random initializations, that sufficiently wide recurrent neural networks satisfy the PL condition. 
Therefore, the PL condition is more applicable, especially in the context of neural networks \citep{madden2020high}.

\begin{assumption}[twice differentiability at the optimum]\label{asm: 2-time diff}
	$f$ is twice continuously differentiable on an open set containing the optimal point $\bx^*$.
\end{assumption}

We make the following assumption on the noise of stochastic gradients, using  $\bw_i^t = \hbg_i^t - \nabla f(\bx_i^t)$ to denote the difference between the stochastic and true gradients.

\begin{assumption}[uniform with strong-growth noise]\label{asm: noise strong growth}
		Conditioned on the iterate $\bx_i^t$, the random variable $\bw_i^t$ is zero-mean and independent with its expected squared norm error bounded as,
	\begin{align*}
	\E[\Vert \bw_i^t \Vert^2|\bx_i^t] \leq c\Vert \nabla f(\bx_i^t) \Vert^2 + \sigma^2,
	\end{align*}
	where $\sigma^2,c\geq0$ are constants.
\end{assumption}
The noise model of Assumption \ref{asm: noise strong growth} is very general and it includes the common case with uniformly bounded squared norm error when $c=0$. As it is noted by \cite{zhang2016parallel}, the advantage of periodic averaging compared to one-shot averaging only appears when $c/\sigma^2$ is large. Therefore, to study Local SGD, it is important to consider a noise model as in Assumption \ref{asm: noise strong growth} to capture the effects of frequent averaging.
Among the related works mentioned in Table \ref{table: comparison}, only \cite{stich2019error} and \cite{haddadpour2019local} analyze this noise model while the rest study the special case with $c=0$.
SGD under this noise model with $c>0$ and $\sigma^2=0$ was first studied in \cite{schmidt2013fast} under the name \textit{strong-growth condition}. Therefore we refer to the noise model considered in this work as \textit{uniform with strong-growth}.

\begin{assumption}[sub-Gaussian noise]\label{asm: noise sub gaussian}
	Conditioned on the iterate $\bx_i^t$, random variable $\bw_i^t$ is zero-mean, independent and $[\bw_i^t]_l$ is $(\sigma/\sqrt{d})$-sub-Gaussian, for $l=1,\ldots,d$, i.e., 
	\[ \E[\exp(\lambda([\bw_i^t]_l - \E[\bw_i^t]_l))| \bx_i^t] \leq \exp\left(\frac{\lambda^2 \sigma^2}{2d} \right), \qquad \forall \lambda \in \R, l=1,\ldots,d.  \]
	Thus, it has uniformly bounded variance $\E[\Vert\bw_i^t \Vert^2| \bx_i^t] \leq \sigma^2$.
\end{assumption}
A sub-Gaussian noise model is commonly assumed for deriving concentration bounds for SGD, which we will use to prove our results for OSA.

As already mentioned in the Introduction, the main goal of this paper is to study the effect of communication times on the convergence of the Local SGD and provide better theoretical guarantees. In what follows, we claim that by carefully choosing the communication times, linear speed-up of parallel SGD can be attained with only a small number of communication instances. Moreover, we will obtain a set of sufficient conditions for OSA to achieve linear speed-up.

\section{Convergence results}\label{sec: convergence}
In this section we present our main convergence results for Local SGD and OSA.
In what follows, we denote by $\bbx^t := (\sum_{i=1}^N \bx_i^t)/N$  the average of the iterates of all workers. Notice that $\bx_i^t = \bbx^t$ for $t\in \I$ and $i=1,\ldots,N$.

\subsection{Local SGD}

Let us introduce the notation $$0=\tau_0 <\tau_1<\ldots<\tau_R = T,$$ for the communication times. Further, let us define $H_i := \tau_{i+1} - \tau_i$ to be the $i$'th interc-communication interval. Our first theorem gives a performance bound under the assumption that $H_i$ grows linearly with $i$. 

\begin{theorem}\label{thm: logT}
	Suppose Assumptions~\ref{asm: smoothness} (smoothness), \ref{asm: strong convexity} (strong convexity) and \ref{asm: noise strong growth} (uniform with strong growth noise) hold.
	
	Choose the parameters as follows:  $R$ 
	such that $1 \leq R \leq \sqrt{2T}$ and  $a:=\lceil 2T/R^2 \rceil \geq 1$, $H_i = a(i+1)$ and $\tau_{i+1} = \min(\tau_i + H_i, T)$ for $i=0,\ldots,R-1$.
	Choose $\beta \geq \max \{9\kappa, 12\kappa^2c \max\{\ln(3), \ln(1+T/(4 \kappa R^2))\} + 3\kappa(1+c/N) \}$ and set the learning rate as $\eta_t = 3/ \mu (t+\beta), t=0,1,\ldots,T-1$. 
	
	Then using Algorithm \ref{alg: Local SGD} we have,
	\begin{align*} 
	\E[f(\bbx^T)] - f^* \leq \frac{ \beta^2 (f(\bbx^0) - f^*)}{T^2}
	+  \frac{9 L \sigma^2}{2 \mu^2 N T} 
	+ \frac{144 L^2 \sigma^2}{\mu^3 RT}.
	\end{align*}
\end{theorem}

\medskip

\begin{corollary}\label{cor: R=N}
    Under the assumptions of Theorem \ref{thm: logT}, selecting the number of communications $R=\Omega(\kappa N)$ we obtain
    \begin{align*}
        \E[f(\bbx^T)] - f^* \leq \frac{ \beta^2 (f(\bbx^0) - f^*)}{T^2}
	+ \O \left( \frac{L \sigma^2}{\mu^2 N T} \right).
    \end{align*}
\end{corollary}

The choice of communication times in Theorem \ref{thm: logT} aligns with the intuition that workers need to communicate more frequently at the beginning of the optimization. As the the step-sizes become smaller and workers' local parameters get closer to the global minimum, they diverge more slowly from each other and therefore, less communication is required to re-align them.
The advantage of this communication strategy over fixed periodic averaging has been only empirically shown in \citet{haddadpour2019local}.
The proof of Theorem \ref{thm: logT} can be found in Appendix~\ref{sec: apx local sgd}.


\subsection{One-shot averaging}
The previous literature literature has shown OSA achieves asymptotic linear speed-up under some restrictive assumptions. For instance, 
\cite{dieuleveut2019communication} show this for three times continuously differentiable functions with second
and third uniformly bounded derivatives. 
Similarly, \cite{godichon2020rates} require the objective function to be strongly convex, twice continuously differentiable almost everywhere, 
with a bounded Hessian everywhere and gradients satisfying the following condition for some constant $C_m$ and all $\bx\in \R^d$,
\begin{align*}
    \left\Vert \nabla f(\bx) - \nabla^2 f(\bx^*)(\bx - \bx^*) \right\Vert \leq C_m \Vert \bx - \bx^* \Vert^2. 
\end{align*}
This inequality is similar to the assumption from \cite{dieuleveut2019communication} of uniformly bounded third derivatives. 
In the following theorem, we relax these assumptions and show that OSA achieves linear speed-up under considerably milder assumptions.

Before proceeding, let us define the step-size sequence $\{\theta_t\}$ as
\begin{align}\label{eq: theta}
\theta_t = \begin{cases}
\frac{1}{L}, \qquad &\text{for } t= 0,\ldots,t_0-1, \\
\frac{2t}{\mu(t+1)^2}, \qquad &\text{for } t\geq t_0,
\end{cases}
\end{align}
where $t_0 = \lfloor 2L/\mu \rfloor$. Notice that $\theta_t \leq 1/L$ for all $t$.

\begin{theorem}\label{thm: osa main}
	Under Assumptions \ref{asm: smoothness} (smoothness), \ref{asm: PL} (PL condition), \ref{asm: 2-time diff} (twice differentiability at the optimum) and \ref{asm: noise sub gaussian} (sub-Gaussian noise) and with step-size sequence $\{ \eta_t\} = \{ \theta_t\}$ defined in \eqref{eq: theta}, we have for $T \geq t_0$,
	\begin{align*}
	\E \left[ \left\Vert \bbx^T - \bx^* \right\Vert^2 \right] \leq 
    \frac{4\sigma^2}{3 \mu^2 NT} + o\left( \frac{1}{T} \right).
	\end{align*}
\end{theorem}

We are thus able to relax the conditions from the earlier literature, which required everywhere or almost everywhere  higher derivatives with uniform bounds on third derivatives to merely twice differentiability at a single point. As a bonus, we also replace strong convexity with the PL condition. 

This theorem is proved in Appendix~\ref{sec: apx osa}. The main difference between Theorem~\ref{thm: osa main} and Corollary \ref{cor: R=N} is that Theorem \ref{thm: osa main} shows a linear speed-up with only one communication round but with slightly more restrictive assumptions such as sub-Gaussian noise model and twice-differentiable objective function at the optimal point. On the other hand, our results for OSA only require the PL-condition instead of strong convexity.

\section{Numerical experiments}\label{sec: Numerical }
To verify our findings and compare different communication strategies in Local SGD, we performed the following numerical experiments, using an Nvidia GTX-1060 GPU and Intel Core i7-7700k processor.

\subsection{Quadratic function with strong-growth condition}
As discussed in \cite{zhang2016parallel,dieuleveut2019communication}, under uniformly bounded variance, one-shot averaging performs asymptotically as well as mini-batch SGD, at least for quadratic functions. Therefore, to fully capture the importance of the choice of communication times $\I$, we design a \textit{hard} problem, where noise variance is uniform with strong-growth condition, defined in Assumption \ref{asm: noise strong growth}.  Let us define,
\begin{align}\label{eq: f-zeta}
F(\bx) = \E_\zeta f(\bx, \zeta), \qquad f(\bx, \zeta) := \sum_{i=1}^d \frac{i}{2}x_i^2 (1 + z_{1,i}) + \bx^\top \bz_2,
\end{align}
where $\zeta = (\bz_1, \bz_2)$ and $\bz_1, \bz_2 \in \R^d, z_{1,i}\sim \mathcal{N}(0,c_1)$ and $z_{2,i} \sim \mathcal{N}(0,c_2)$, $\forall i \in[d]$, are random variables with normal distributions. We assume at each iteration $t$, each worker $i$ samples a $\zeta_i^t$ and uses $\nabla f(\bx, \zeta_i^t)$ as a stochastic estimate of $\nabla F(\bx)$.  It is easy to verify that $F(\bx)$ is $1$-strongly convex and $d$-smooth, $F^* = 0$ and $\E_\zeta[\Vert \nabla f(\bx, \zeta) - \nabla F(\bx) \Vert^2] = c \Vert \nabla F(\bx) \Vert^2 + \sigma^2$, where $c = c_1$ and $\sigma^2 = dc_2$.

We use Local SGD to minimize $F(\bx)$ using different communication strategies, namely, synchronized SGD where $H=1$,  $H \approx \sqrt{TN}$ \cite{stich2018local}, $H \approx (TN)^{1/3}$ \cite{haddadpour2019local}, $R=N$ with constant $H \approx T/N$ \cite{stich2019error,khaled2019tighter} and finally the communication strategy proposed in this work with $R=N$ and linearly growing $H_i$ local steps.
We used $N=20$ workers, $T=1000$ iterations, $c_1 = 1.0$ and $c_2= 10^{-10}$ with $d=3$ and step-size sequence $\eta_t = 3/(\mu(t+1))$. To estimate the expected value of errors, we repeated the optimization using each strategy $100$ times and reported the average and $1$-standard-deviation error bar in Figure~\ref{fig: quadratic}.

\begin{figure}
	\centering
	\begin{subfigure}{0.5\textwidth}
		\includegraphics[height=5cm]{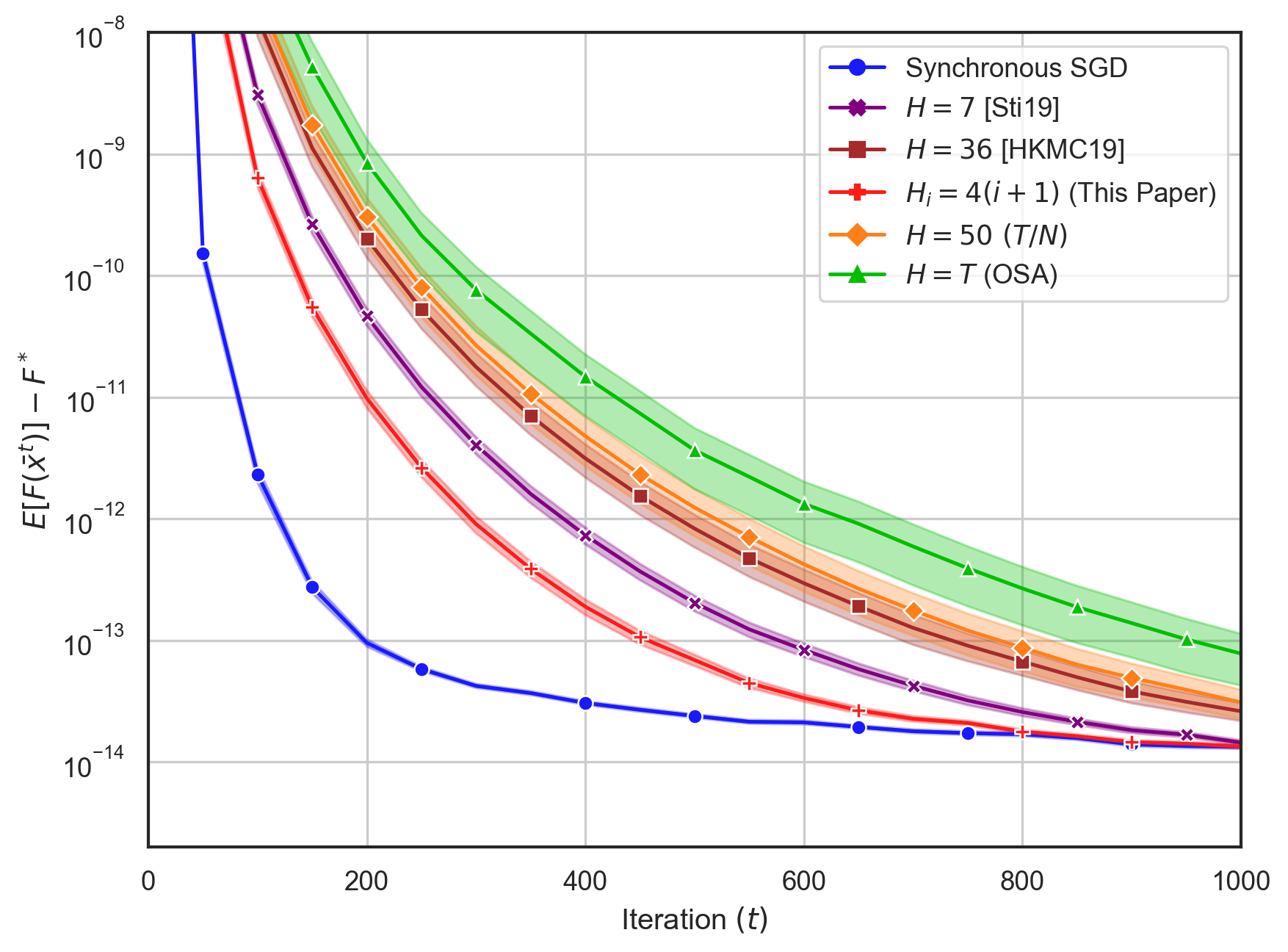}
		\caption{Error over iteration.}
	\end{subfigure}
	\begin{subfigure}{0.49\textwidth}
		\includegraphics[height=5cm]{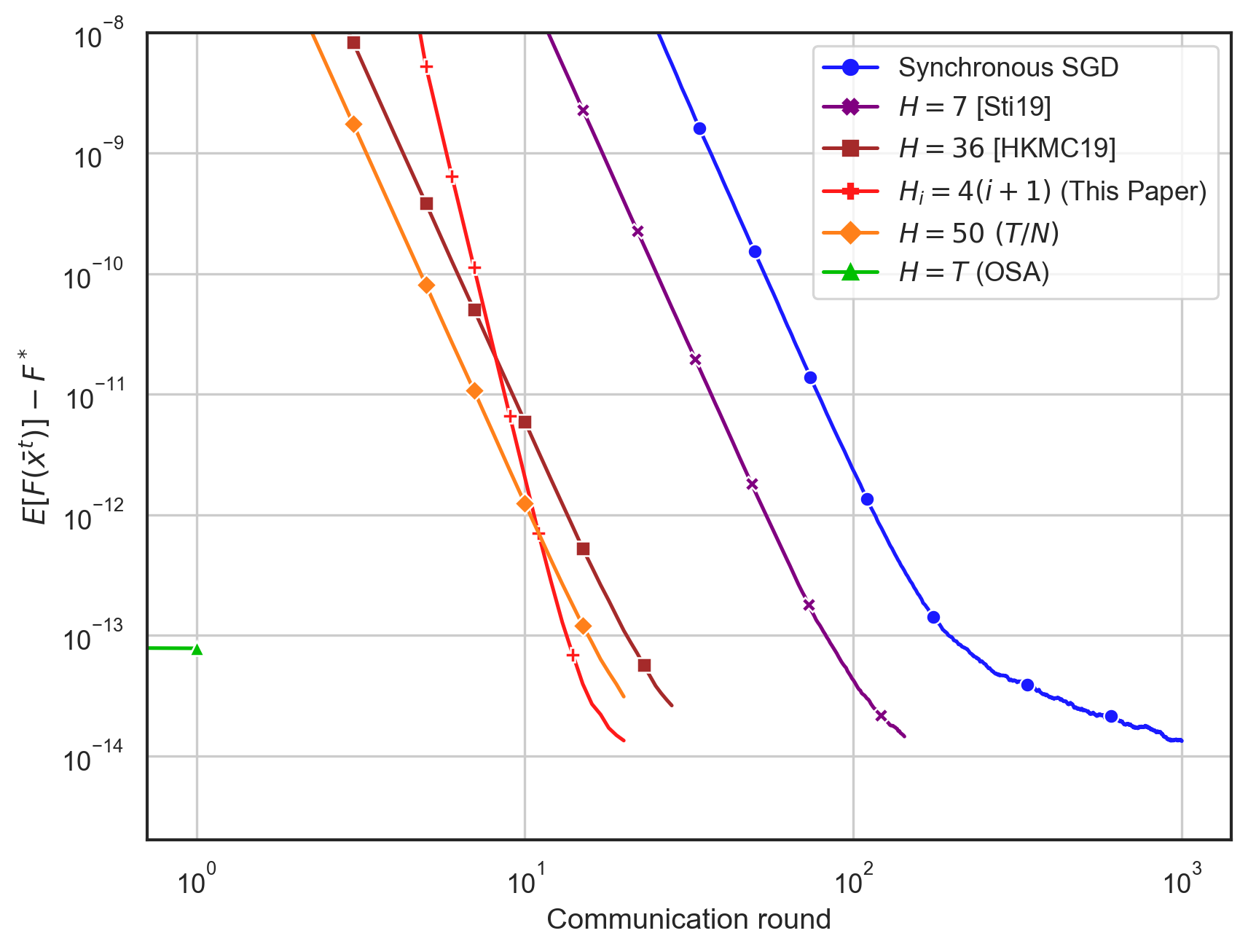}
		\caption{Error over communication round.}
	\end{subfigure}
	\caption{Minimizing \eqref{eq: f-zeta} using Local SGD with different communication strategies. Figures (a) and (b) show the error over iteration and communication rounds, respectively.}
	\label{fig: quadratic}
\end{figure}

We make the following observations from Figure~\ref{fig: quadratic}:
\begin{itemize}
	\item
	Figure~\ref{fig: quadratic}(a) shows that a communication strategy with increasing local steps (proposed in this work), outperforms all the other methods, both in transient and final error performance, specifically the one with the same number of communication rounds evenly spread throughout the whole optimization.
	This confirms the advantage of more frequent communication at the beginning of the optimization, especially when the ratio of $c$ to $\sigma^2$ in the noise with growth condition is large (see the definition in Assumption \ref{asm: noise strong growth}).
	\item
	Figure\ref{fig: quadratic}(b) shows that our communication method uses fewer communication rounds, $20$ versus $28$ \citep{haddadpour2019local}, $143$ \citep{stich2018local} and $1000$ rounds for synchronized SGD.
	\item
	OSA appears to perform relatively well despite using only one communication round, 
	though not quite as well as other methods. This shows that the choice of communication is important in this experiment. In other words, it is not true that the success of our communication strategy is  merely a byproduct of the experiment design, where any communication strategy, as long as it communicates at least once, will succeed.
\end{itemize}


\subsection{Speed-up curves}
In this experiment, we minimize a one-dimensional function defined as,
\begin{align}\label{eq: osa fail}
	F(x) = \begin{cases}
	\frac{1}{2}x^2, \qquad &x\leq 0,\\
	x^2, \qquad &x>0,
	\end{cases}
\end{align}
using Local SGD with gradients corrupted by a normal noise $\mathcal{N}(0,\sigma^2)$. We chose this specific cost function since it is not twice continuously differentiable at the minimizer $x^*=0$ and does not satisfy Assumption \ref{asm: 2-time diff} required by Theorem \ref{thm: osa main} for OSA to achieve linear speed-up. The results of this experiment will help us understand whether twice differentiability is a necessary assumption for OSA to obtain a linear speed-up.

The speed-up curve is derived by dividing the \emph{expected} error of a single worker SGD by the \emph{expected} error of each method at the final iterate $T$, over different number of workers $N$. Thus in the case where the error decreases linearly in the number of workers, we should expect to see a straight line on the graph. 

We plot the speed-up curve for $N$ workers using different communication strategies: synchronized SGD, $R=N$ communication rounds with linearly increasing number of local steps $H_i$, $R=N$ with constant number of local steps $H\approx T/R$, as well as OSA with only $R=1$ communication at the end. 
We use the step-size sequence $\eta_t = \min\{1/L, 2/(\mu (t+1))\}$ with $\mu=1, L=2$, and $\sigma=8$, $T=1000$.

\begin{figure}
	\centering
	\begin{subfigure}{0.45\textwidth}
		\includegraphics[height=5cm]{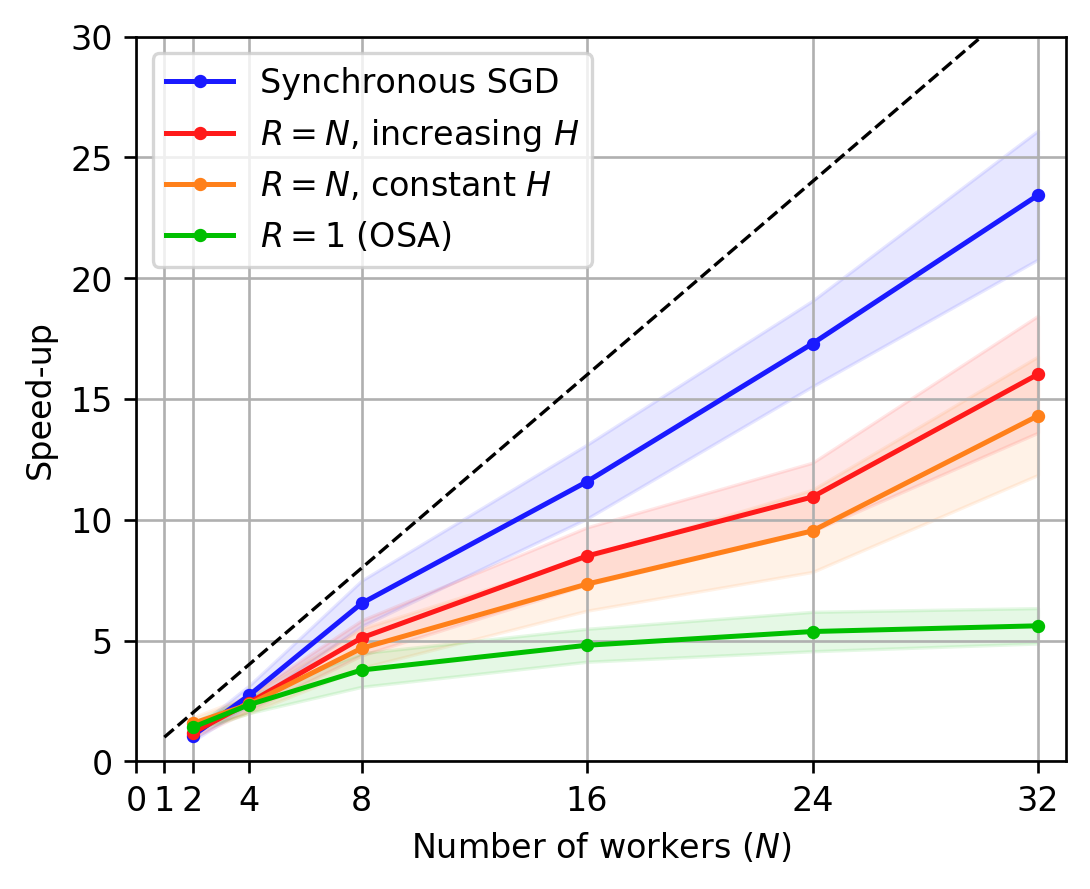}
		\caption{}
	\end{subfigure}
	\begin{subfigure}{0.53\textwidth}
		\includegraphics[height=5cm]{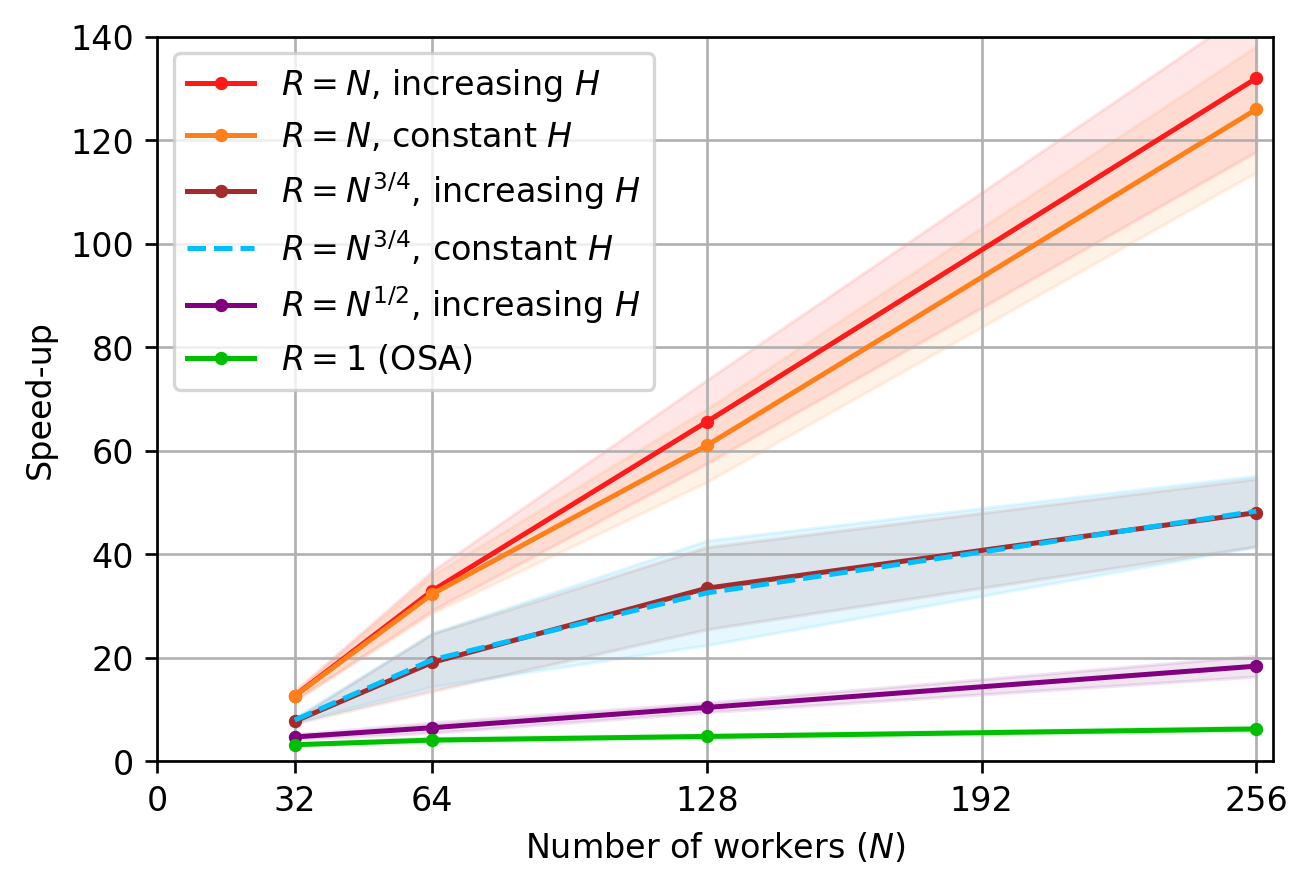}
		\caption{}
	\end{subfigure}
	\caption{Speed-up curves for different communication strategies, over different ranges of $N$ and $T$. Figure (a) establishes the linear speed-up of local SGD with $R=N$ communication rounds as well as failure of OSA to achieve speed-up even with small number of workers $N\leq 32$ over $T=1000$ iterations. Figure (b) additionally plots speed-up curves for $R \approx N^{3/4}$ and $R \approx N^{1/2}$ for larger values of $32\leq N \leq 256$ and $T=8000$.}
	\label{fig: speed-up}
\end{figure}

Our results in Figure~ \ref{fig: speed-up}(a) show that Local SGD with $R=N$ (increasing or constant $H$) achieves linear speed-up in the number of workers, albeit with a worse constant compared to synchronized SGD. However, OSA fails to scale as $N$ increases. This suggests that the condition of  twice differentiability (Assumption~\ref{asm: 2-time diff}) is necessary for Theorem \ref{thm: osa main}, as this function satisfies all the other assumptions of that theorem.

While our theoretical results provide only an upper bound on $R$ to achieve linear speed-up, this setting gives us a chance to find out if smaller number of communication rounds are enough. Therefore we repeat this experiment for larger number of workers $N$ and $T=8000$, using $R \approx N^{3/4}$ and $R \approx N^{1/2}$ communication rounds.  
Our results in Figure~\ref{fig: speed-up}(b) show that $R=N$ clearly achieves speed-up for larger values of $N$, as expected and $R=1$ and $R\approx N^{1/2}$ fail to speed-up. However, $R \approx N^{3/4}$ also struggles to \emph{linearly} speed-up in the number of workers, as the slope of the speed-up curve declines with $N$ increasing. It would be of interest to look into a more granular choice of communication rounds such as $R \approx N^{0.9}$ or even $R \approx N^{0.99}$  but this would require much larger values of $N$ and $T$ and thus more repeated simulations, which is beyond our computational resources, which were already exhausted by generating Figure 2(b).

It is worth mentioning that in both experiments of Figure \ref{fig: speed-up}(a) and \ref{fig: speed-up}(b), $R=N$ with increasing $H$ outperforms the one with constant $H$, even though the noise model used in this experiment is simply uniformly bounded, without strong-growth condition. This further endorses the use of more frequent averaging at the beginning of optimization, when paired with decreasing step-size sequence.

\subsection{Regularized logistic regression}
We also performed additional numerical experiments with regularized logistic regression using two large real datasets: (i) a national dataset (NSQIP) of surgeries performed in the U.S., seeking to predict short-term hospital re-admissions, which consists of  $\bf {722101}$ data points (surgeries) each characterized by $d=231$ features, (ii) the a9a dataset from 
LIBSVM \citep{CC01a} which includes ${\bf 32561}$ data points with $d = 124$ features. The results of these experiments are presented and discussed in Appendix~\ref{sec: apx numerical}.

\section{Conclusion}\label{sec: conclusion}
In this work, we studied the communication complexity of Local SGD and provided an analysis that shows that $R=\Omega(N)$ number of communication rounds, independent of the total number of iterations $T$, is sufficient to achieve linear speed-up. Moreover, we showed only a single round of averaging is needed provided that the objective is twice differentiable at the optimum point. This assumption appears to be necessary, as our simulations show that not only one-shot averaging but using $N^{1/2}$ or $N^{3/4}$ communications in local SGD fails to deliver linear speed-up on a simple example which is not twice differentiable at the optimum.

\newpage
\begin{ack}
The research was partially supported by the NSF under grants DMS-1664644, CNS-1645681, ECCS-1933027
and IIS-1914792, by the ONR under grants N00014-19-1-2571 and N00014-21-1-2844, by the ARO under grant W911NF-1-1-0072, 
by the NIH under grants R01
GM135930 and UL54 TR004130, by the DOE under grants DE-AR-0001282 and NETL-EE0009696, and by the Boston University  Kilachand Fund
for Integrated Life Science and Engineering.

\end{ack}


\bibliographystyle{rusnat}
\bibliography{References}


\newpage
\section*{Checklist}


\begin{enumerate}

\item For all authors...
\begin{enumerate}
  \item Do the main claims made in the abstract and introduction accurately reflect the paper's contributions and scope?
    \answerYes{}
  \item Did you describe the limitations of your work?
    \answerYes{}
  \item Did you discuss any potential negative societal impacts of your work?
    \answerNA{}
  \item Have you read the ethics review guidelines and ensured that your paper conforms to them?
    \answerYes{}
\end{enumerate}

\item If you are including theoretical results...
\begin{enumerate}
  \item Did you state the full set of assumptions of all theoretical results?
    \answerYes{}
	\item Did you include complete proofs of all theoretical results?
    \answerYes{See Appendix}
\end{enumerate}

\item If you ran experiments...
\begin{enumerate}
  \item Did you include the code, data, and instructions needed to reproduce the main experimental results (either in the supplemental material or as a URL)?
    \answerYes{}
  \item Did you specify all the training details (e.g., data splits, hyperparameters, how they were chosen)?
    \answerYes{}
	\item Did you report error bars (e.g., with respect to the random seed after running experiments multiple times)?
    \answerYes{}
	\item Did you include the total amount of compute and the type of resources used (e.g., type of GPUs, internal cluster, or cloud provider)?
    \answerYes{}
\end{enumerate}

\item If you are using existing assets (e.g., code, data, models) or curating/releasing new assets...
\begin{enumerate}
  \item If your work uses existing assets, did you cite the creators?
    \answerYes{}
  \item Did you mention the license of the assets?
    \answerNA{}
  \item Did you include any new assets either in the supplemental material or as a URL?
    \answerYes{}
  \item Did you discuss whether and how consent was obtained from people whose data you're using/curating?
    \answerNA{}
  \item Did you discuss whether the data you are using/curating contains personally identifiable information or offensive content?
    \answerNA{}
\end{enumerate}

\item If you used crowdsourcing or conducted research with human subjects...
\begin{enumerate}
  \item Did you include the full text of instructions given to participants and screenshots, if applicable?
    \answerNA{}
  \item Did you describe any potential participant risks, with links to Institutional Review Board (IRB) approvals, if applicable?
    \answerNA{}
  \item Did you include the estimated hourly wage paid to participants and the total amount spent on participant compensation?
    \answerNA{}
\end{enumerate}

\end{enumerate}


\newpage
\appendix
\section{More numerical experiments}\label{sec: apx numerical}
In this section we present additional numerical experiments.
We consider binary classification and select $l_2$-regularized logistic regression with its corresponding loss function as the objective function $F$ to be minimized, i.e.,
\begin{align}\label{eq: logistic loss}
F(\bx) = \frac{1}{M} \sum_{j=1}^M \left( \ln(1+\exp(\bx^\top \mathbf A_j)) - 1_{(b_j = 1)} \bx^\top \mathbf A_j \right) 
+ \frac{\lambda}{2} \Vert \bx \Vert_2^2,
\end{align}
where $\lambda$ is the regularization parameter, $\mathbf A_j \in \R^d$ and $b_j \in \{ 0,1 \}$, $j=1,\ldots,M$ are features (data points) and their corresponding class labels, respectively.

\subsection{Fixed number of workers}
Here we use two large datasets. 
One, a real dataset from the American College of Surgeons National Surgical Quality Improvement Program (NSQIP) to predict whether a specific patient will be re-admitted within 30 days from discharge after general surgery. This dataset consists of $M = \num{722101}$ data points for training with $d=231$ features including (i) baseline demographic and healthcare status characteristics, (ii) procedure information and (iii) pre-operative, intra-operative, and post-operative variables.
Second, the a9a dataset from LIBSVM \citep{CC01a}. This dataset consists of $M=\num{32561}$ data points for training with $d = 124$ features.

We perform Local SGD with $N=10$ workers, $\lambda = 0.05$, step-size sequence $\eta_t = 3/(\mu(t+1))$ ($\beta=1$), $T=1000$ iterations and batch size of $b=1$  with different communication strategies:
(i) synchronized SGD with $H=1$,
(ii) a strategy with the time varying communication intervals with $H_i = a(i+1), a \approx 18$ and $R = 10$ communication rounds proposed in this paper,
(iii) a strategy with the same number of communications however with a fixed $H=T/N = 100$, and finally, (iv) one-shot averaging with $H=T$.
Each simulation has been repeated $10$ times and the average of their performance is reported in Figure \ref{fig: logistic}.

It can be seen from Figure~\ref{fig: logistic} that all of the communication methods, including OSA, have similar terminal error as synchronized SGD. This further validated our results, especially Theorem \ref{thm: osa main}, since the logistic loss is both twice differentiable and satisfies the PL condition, due to strong convexity of the $l_2$-regularization. Moreover, we do not notice any significant difference between the performance of the varying and constant local steps, mainly because even a method with only one communication round (OSA) performs just as well.

\begin{figure}
	\centering
	\begin{subfigure}{0.49\textwidth}
		\includegraphics[width=6.5cm]{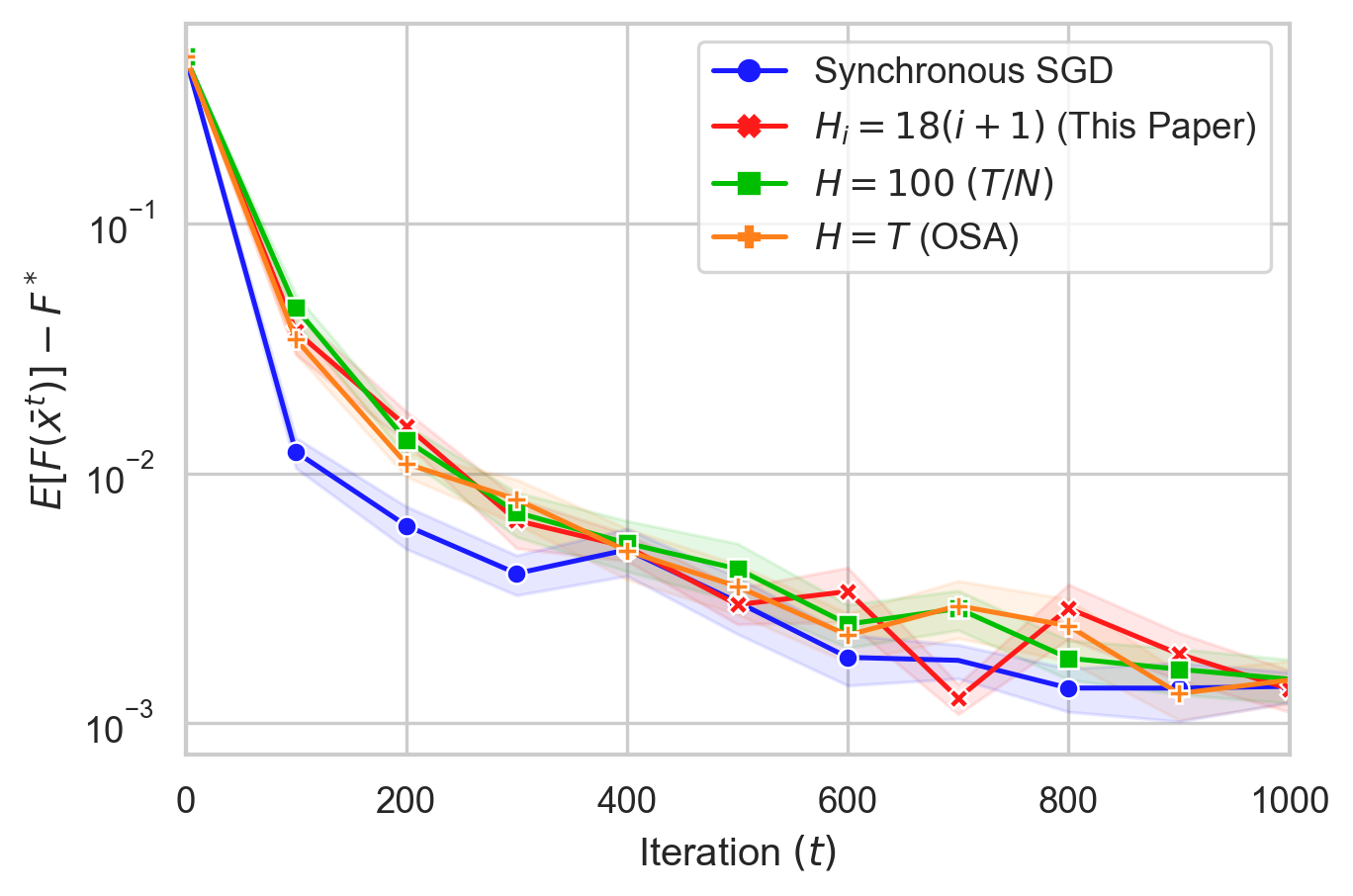}
		\caption{NSQIP data set.}
	\end{subfigure}
	\begin{subfigure}{0.49\textwidth}
		\includegraphics[width=6.5cm]{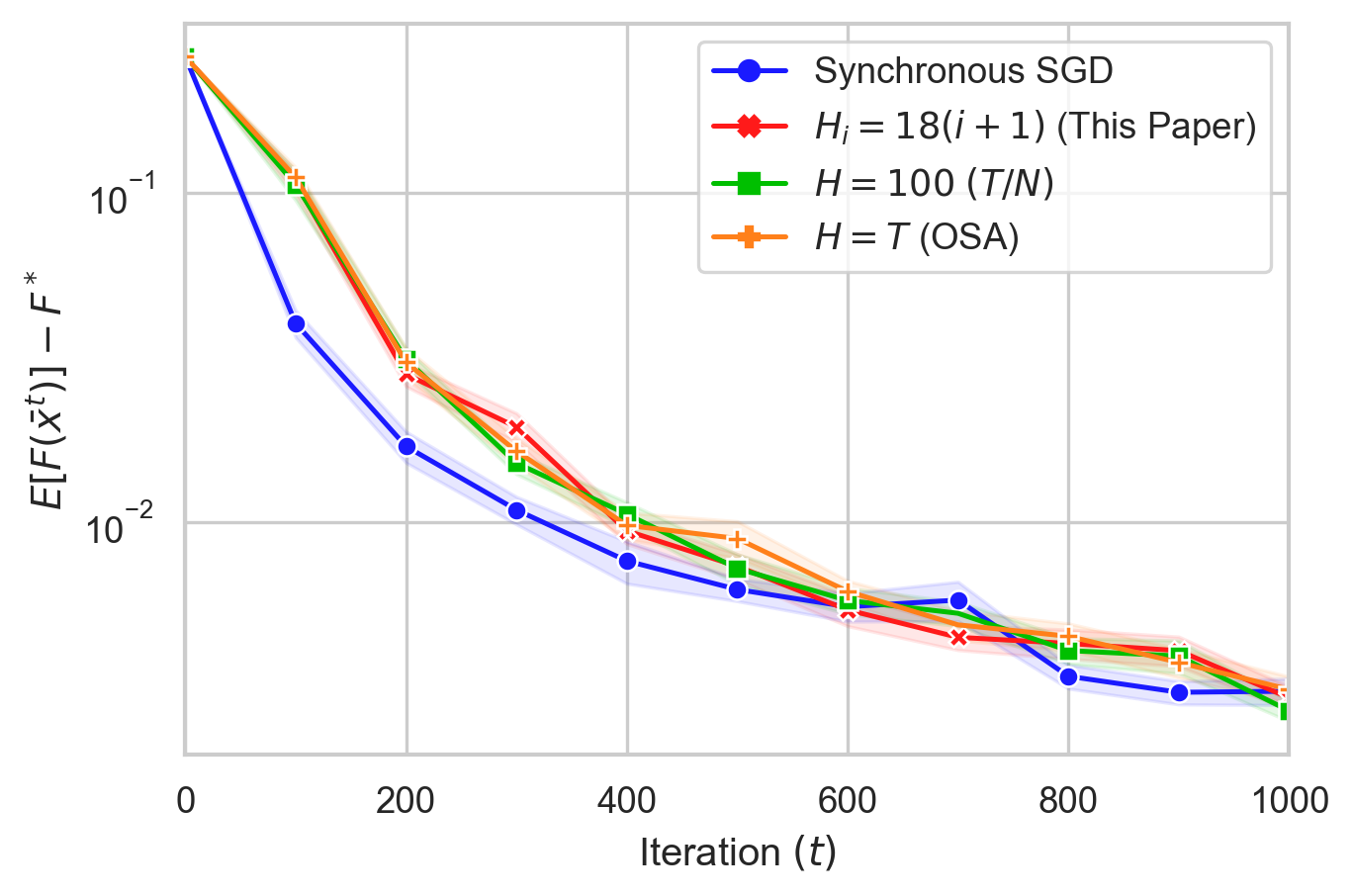}
		\caption{a9a data set.}
	\end{subfigure}
	\caption{Minimizing \eqref{eq: logistic loss} using Local SGD with different communication strategies. Figures (a) and (b) show the error over iteration for NSQIP and a9a datasets, respectively. The shaded areas show the $1$-standard deviation error bar.}
	\label{fig: logistic}
\end{figure}

\newpage
\subsection{Comparison with FedAC}

Here, we perform an extensive comparison between different methods using different number of workers $N$ and communication rounds $R$. We adopt a setting similar to that of Figure 4 in \citet{FedAC}. More specifically, we compare our communication strategy with other baselines and FedAC, using logistic regression \eqref{eq: logistic loss} on the a9a dataset with $\lambda=0.01$ and $T=8192$.

The results in Figures \ref{fig:FedAC over N} and \ref{fig:FedAC over H} are obtained by tuning the fixed learning rate $\eta$ over the set $\{1e^{-3}, 2e^{-3}, 5e^{-3}, 1e^{-2}, \ldots, 2, 5, 10 \}$ for all the methods except for Local SGD with growing intervals, where we used $\eta_t = 3/(\mu(t+1))$ without any tuning.

\begin{figure}[h]
    \centering
    \includegraphics[width=\textwidth]{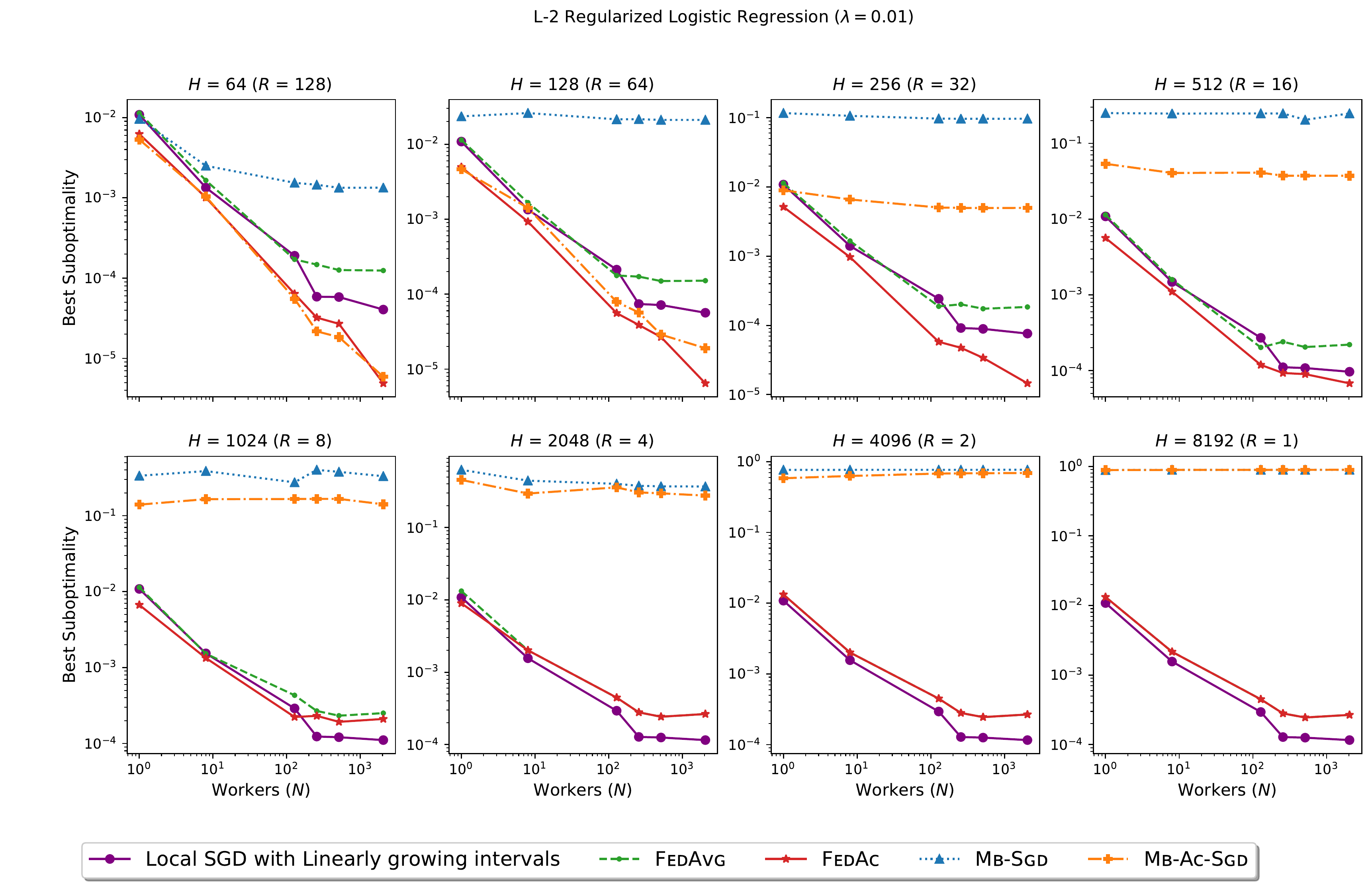}
    \caption{Comparison of Local SGD with (linearly) growing communication intervals introduced in this paper with other baseline methods on the observed linear speed-up w.r.t. $N$ workers ($\lambda=0.01$).}
    \label{fig:FedAC over N}
\end{figure}

We observe from Figure \ref{fig:FedAC over N} that when the number of communications $R$ is large ($R\geq 16$), FedAC has better performance across different values of $N$.
However, as the communication becomes sparse, Local SGD with growing communication intervals outperforms all the other methods, specifically as the number of workers increases. 
We also notice that both Mini-Batch SGD and its accelerated version have a relatively poor performance as $N$ or $H$ increase.
Similar observations can be made from Figure \ref{fig:FedAC over H}.

\begin{figure}
    \centering
    \includegraphics[width=0.9\textwidth]{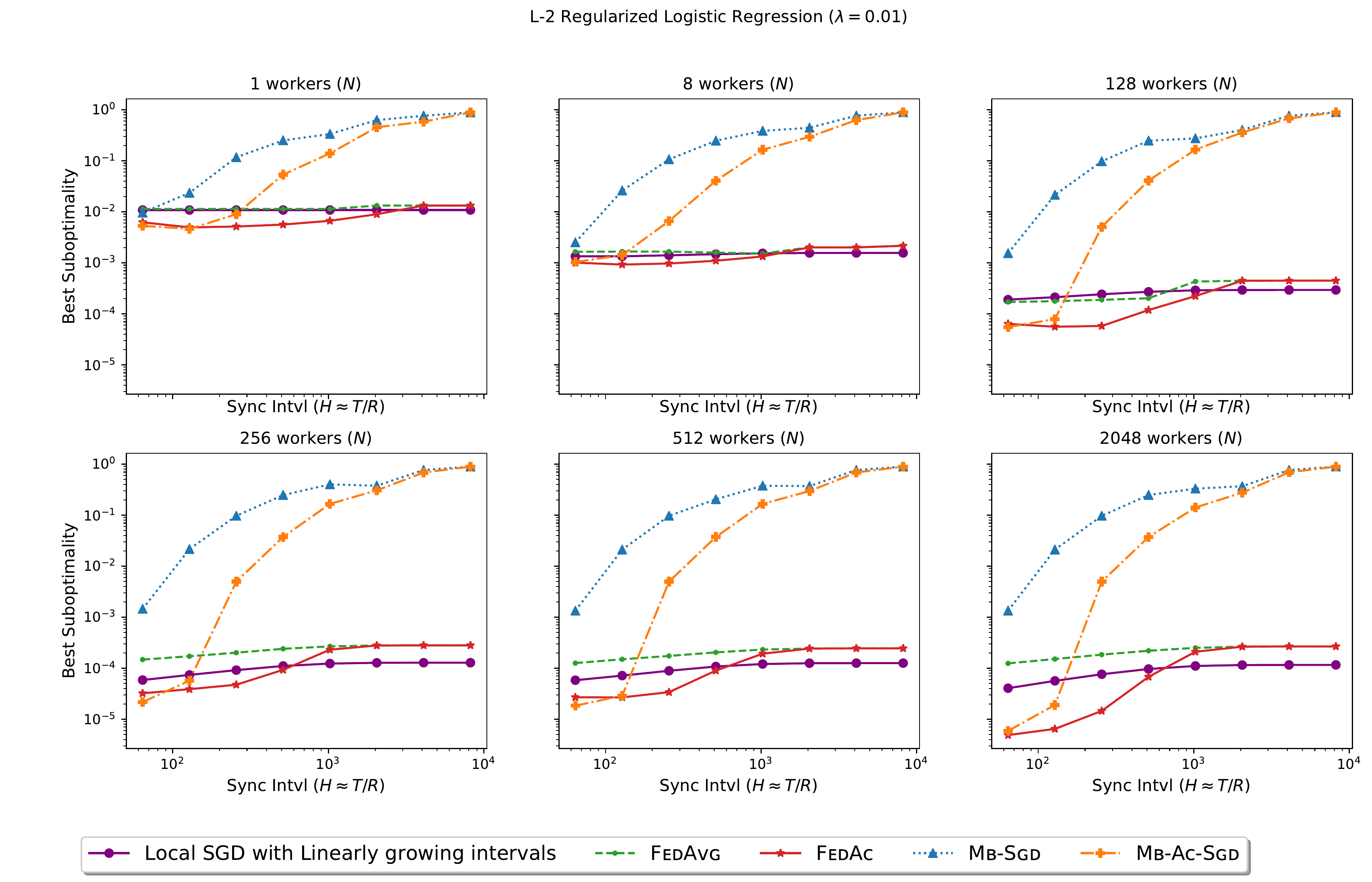}
    \caption{Comparison of Local SGD with (linearly) growing communication intervals introduced in this paper with other baseline methods on the dependency on number of communications ($\lambda=0.01$).}
    \label{fig:FedAC over H}
\end{figure}

\newpage
We notice that increasing strong convexity to $\lambda=1.0$, results in our communication strategy to uniformly outperform all the other methods, across all values of $N$ and $R$ (see Figure \ref{fig:FedAC lambda 1}).

\begin{figure}[h]
    \centering
    \includegraphics[width=0.9\textwidth]{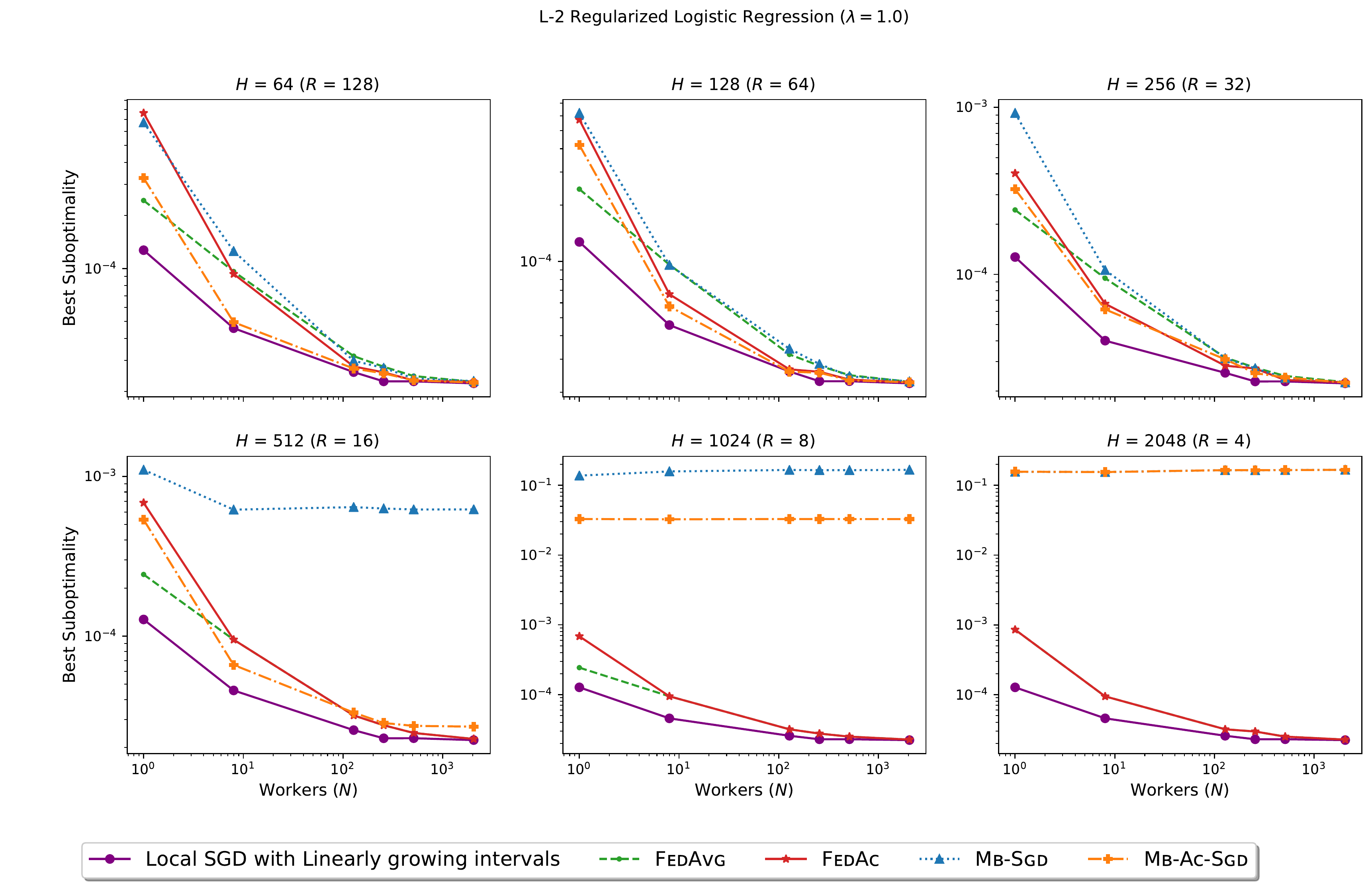}
    \caption{Comparison of Local SGD with (linearly) growing communication intervals introduced in this paper with other baseline methods on the observed linear speed-up w.r.t. $N$ workers ($\lambda=1.0$).}
    \label{fig:FedAC lambda 1}
\end{figure}

\newpage
\section{Local SGD}\label{sec: apx local sgd}
Here we present a few results which will be used later to prove Theorem~\ref{thm: logT} as well as to better understand the choice of varying number of local steps.
In the following theorem, we show an upper bound for the sub-optimality error, in the sense of function value, for any choice of communication times $\I$.
Theorem \ref{thm: logT} will be obtained by specializing the following bound. 

First, let us introduce some notation. Let $0=\tau_0<\tau_1 < \ldots < \tau_R = T$ be the communication times and denote the most recent communication time by $\tau(t) := \max\{t' \in \mathcal{I}| t' \leq t \}$. Define $H_i := \tau_{i+1} - \tau_i$, as the length of the $(i+1)$-st inter-communication interval, for $i=0,\ldots,R-1$. 

\begin{theorem}\label{thm: general}
	Suppose Assumptions~\ref{asm: smoothness}, \ref{asm: strong convexity} and \ref{asm: noise strong growth} hold. Choose $\beta \geq 9\kappa $ and communication times $\I = \{\tau_i|i=1,\ldots,R\}$ such that it holds for $i=0,\ldots,R-1,$
	\begin{align}\label{eq: beta condition}
	12\kappa^2c \ln(1 + \frac{H_i - 1}{\tau_i + \beta}) + 3\kappa(1+\frac{c}{N}) - (\tau_i + \beta)\leq 0.
	\end{align}
	Set step-sizes $\eta_t = 3/ (\mu (t+\beta))$, $t=0,1,\ldots,T-1$. Then, using Algorithm \ref{alg: Local SGD}, we have 
	\begin{align}\label{eq: opt E3}
	\E[f(\bbx^T)] - f^* \leq \frac{\beta^2 (f(\bbx^0) - f^*)}{T^2}  +  \frac{9L \sigma^2}{2 \mu^2 N T} 
	+ \frac{18L^2 \sigma^2}{ \mu^3 T^2} \sum_{t=0}^{T-1} \frac{t-\tau(t)}{t+\beta},
	\end{align}
\end{theorem}
The last term in Equation \eqref{eq: opt E3} is due the to disagreement between workers (consensus error), introduced by local computations without any communication. As the inter-communication intervals become larger, $t-\tau(t)$ becomes larger as well and increases the overall optimization error. This term explains the trade-off between communication efficiency and the optimization error.

Note that condition \eqref{eq: beta condition} is mild. For instance, it suffices to set $\beta \geq \max \{12\kappa^2 c \ln(1 + T/(9\kappa)) + 3\kappa(1+c/N), 9\kappa\}$. Moreover, the bound in \eqref{eq: opt E3} is for the last iterate $T$, and does not require keeping track of a weighted average of all the iterates.

Theorem \ref{thm: general} not only bounds the optimization error, but introduces a methodological approach to select the communication times to achieve smaller errors. For the scenarios when the user can afford to have a certain number of a communications, they can select $\tau_i$ to minimize the last term in \eqref{eq: opt E3}.

\paragraph{One-shot averaging.} Plugging $H=T$ in Theorem~\ref{thm: general}, we obtain a convergence rate of $\O(\kappa^2 \sigma^2/(\mu T))$ without any linear speed-up. Among previous works, only \cite{khaled2019tighter} show a similar result.

\subsection{Fixed-length intervals}
A simple way to select the communication times $\I$, is to split the whole training time $T$ to $R$ intervals of length at most $H$. Then we can use the following bound in Equation \eqref{eq: opt E3}, 
\begin{align*}
\sum_{t=0}^{T-1} \frac{t-\tau(t)}{t+\beta} \leq (H-1)\sum_{t=0}^{T-1} \frac{1}{t+\beta} \leq (H-1) \ln(1+\frac{T}{ \beta - 1}).
\end{align*}
We state this result formally in the following corollary.
\begin{corollary}
	Suppose assumptions of Theorem~\ref{thm: general} hold and in addition, workers communicate at least once every $H$ iterations. Then,
	\begin{align}\label{eq: fixed-int}
	\E[f(\bbx^T)] - f^* \leq \frac{\beta^2(f(\bbx^0) - F^*)}{T^2} +  \frac{9L \sigma^2}{2\mu^2 N T} 
	+ \frac{18 L^2 \sigma^2 (H-1)}{ \mu^3 T^2} \ln(1 + \frac{T}{\beta - 1}).
	\end{align}
\end{corollary}

\paragraph{Linear speed-up.}
Setting $H= \O(T/(N\ln(T)))$ we achieve linear speed-up in the number of workers, which is equivalent to a communication complexity of $R = \Omega(N \ln(T))$. To the best of the authors' knowledge, this is the tightest communication complexity that is shown to achieve linear speed-up. \cite{khaled2019tighter} and \cite{stich2019error} have shown a similar communication complexity.


\paragraph{Recovering synchronized SGD.} When $H=1$, the last term in \eqref{eq: fixed-int} disappears and we recover the convergence rate of parallel SGD, albeit, with a worse dependence on $\kappa$.

\subsection{Sketch of proof}\label{sec: sketch of proof}
Here we give an outline of the proofs for the Local SGD results presented in this paper. The proof of the following lemmas are provided in the next section. 

\paragraph{Perturbed iterates.}
A common approach in analyzing parallel algorithms such as Local SGD is to study the evolution of the sequence $\{\bbx^t\}_{t\geq0}$. We have, 
\begin{align}\label{eq: average x update}
\bbx^{t+1} = \bbx^t - \frac{\eta_t}{N}\sum_{i=1}^N \hbg_i^t 
= \bbx^t - \eta_t \tbg^t,
\end{align}
where $\tbg^t := (\sum_{i=1}^N \hbg_i^t)/N$ is the average of the stochastic gradient estimates of all workers. 

Let us define $\xi^t:= \E[f(\bbx^t)] - f^*$ to be the optimality error. The following lemma, which is similar to a part of the proof found in \cite{haddadpour2019local}, bounds the optimality error at each iteration recursively.
\begin{lemma}\label{lem: error decay}
	Let Assumptions \ref{asm: smoothness}, \ref{asm: strong convexity} and \ref{asm: noise strong growth} hold. Then,
	\begin{align*}
	\xi^{t+1} \leq \xi^t(1 - \mu \eta_t) + \frac{L^2 \eta_t}{2N} \E \left[ \sum_{i=1}^N \Vert \bbx^t - \bx_i^t \Vert^2 \right] 
	+ \frac{\eta_t^2 L}{2} \E[\Vert \tbg^t \Vert_2^2 ] 
	- \frac{\eta_t}{2N} \E \left[ \sum_{i=1}^N \Vert \nabla f(\bx_i^t) \Vert^2\right].
	\end{align*}
\end{lemma}
Equipped with Lemma \ref{lem: error decay}, we can bound the consensus error ($\E[\sum_{i=1}^N \Vert \bbx^t - \bx_i^t \Vert^2]$) as well as the term $\E[\Vert \tbg^t \Vert^2]$ in the following lemmas. 

\paragraph{Consensus error.}
In the following lemmas, we utilize the structure of the problem to bound the consensus error recursively. Let us define $\bg_i^t = \nabla f(\bx_i^t)$ as the true gradient at worker $i$'s iterate at time $t$.
\begin{lemma}\label{lem: consensus 1}
	Let Assumptions \ref{asm: smoothness}, \ref{asm: strong convexity} and \ref{asm: noise strong growth} hold. Then,
	\begin{multline}\label{eq: consensus 2}
	\E\left[ \sum_{i=1}^N \Vert \bx_i^{t+1} - \bbx^{t+1} \Vert^2 \right] \leq  
	\E\left[ \sum_{i=1}^N \Vert \bx_i^{t} - \bbx^t \Vert^2 \right](1 - \eta_t \mu + \eta_t^2 \mu L) \\
	+ (N-1)\eta_t^2\sigma^2 + \left(1-\frac{1}{N} \right)\eta_t^2 c\E \left[\sum_{i=1}^N \Vert \bg_i^t \Vert^2 \right].
	\end{multline}
\end{lemma}

This lemma, bounds how much the consensus error grows at each iteration. Of course, when workers communicate, this error resets to zero and thus, we can calculate an upper bound for the consensus error, knowing the last iteration communication occurred and the step-size sequence. The following lemma takes care of that. Before stating the following lemma, let us define $G^t := \frac{1}{n}\sum_{i=1}^N \Vert \bg_i^t \Vert^2$.

\begin{lemma}\label{lem: consensus 2}
	Let assumptions of Theorem \ref{thm: general} hold. Then,
	\begin{align}
	\E\left[ \sum_{i=1}^N \Vert \bx_i^{t} - \bbx^t \Vert^2 \right] \leq
	12(N-1)\sum_{k=\tau(t)}^{t-1}\frac{c\E \left[ G^k \right] + \sigma^2}{\mu^2(t+\beta)^2}.
	\end{align}
\end{lemma}

\paragraph{Variance.}
Our next lemma bounds $E[\Vert \tbg^t \Vert^2]$.
\begin{lemma}\label{lem: variance}
	Under Assumption \ref{asm: strong convexity} we have,
	\begin{align*}
	\E \left[ \left\Vert \tbg^t \right\Vert^2 \right] \leq \left(1+\frac{c}{N} \right) 
	\E \left[G^t \right] + \frac{\sigma^2}{N}.
	\end{align*}
\end{lemma}

\subsection{Proofs}
Let us define the following notations used in the proofs presented here.
\begin{align*}
\bbg^t := \frac{1}{N} \sum_{i=1}^n \bg_i^t, \qquad G^t := \frac{1}{N}\sum_{i=1}^N \Vert \bg_i^t \Vert^2, \qquad & \bw_i^t := \hbg_i^t - \bg_i^t.
\end{align*}
Moreover, define $\F^t:=\{\bx_i^k, \hbg_i^k | 1\leq i \leq N, 0\leq k \leq t-1 \} \cup \{\bx_i^t | 1 \leq i \leq N\}$.

\begin{proof}[Proof of Lemma \ref{lem: error decay}]
	By Assumptions \ref{asm: smoothness} and \ref{asm: strong convexity} and \eqref{eq: average x update} we have,
	\begin{align}\label{eq: opt E1}
	\E[ f(\bbx^{t+1}) - f(\bbx^t)] \leq -\eta_t \E[\langle \nabla f(\bbx^t), \tbg^t \rangle ] + \frac{\eta_t^2 L}{2} \E[\Vert \tbg^t \Vert_2^2 ].
	\end{align}
	We bound the first term on the R.H.S of \eqref{eq: opt E1} by conditioning on $\F^t$ as follows:
	\begin{align} \label{eq: opt E2}
	\E[\langle  \nabla f(\bbx^t), \tbg^t \rangle | \mathcal{F}^{t} ] &= \frac{1}{N} \sum_{i=1}^N \langle \nabla f(\bbx^t), \E[\hbg_i^t | \bx_i^t] \rangle \nonumber \\
	&= \frac{1}{2} \Vert \nabla f(\bbx^t) \Vert^2 + \frac{1}{2N} \sum_{i=1}^N \Vert \nabla f(\bx_i^t) \Vert^2 - \frac{1}{2N} \sum_{i=1}^N \Vert \nabla f(\bbx^t) - \nabla f(\bx_i^t) \Vert^2  \nonumber\\
	& \geq \mu (f(\bbx^t) - f^*) + \frac{1}{2N} \sum_{i=1}^N \Vert \nabla f(\bx_i^t) \Vert^2  - \frac{L^2}{2N} \sum_{i=1}^N \Vert \bbx^t - \bx_i^t \Vert^2,
	\end{align}
	where we used $\langle a, b \rangle  = \frac{1}{2} \Vert a \Vert ^2 + \frac{1}{2} \Vert b\Vert^2 - \frac{1}{2} \Vert a-b \Vert ^2$ in the second equation and $(1/2)\Vert \nabla f(\bx) \Vert^2 \geq \mu (f(\bx) - f^*)$ as well as smoothness of $f$ in the last inequality.
	Taking full expectation of \eqref{eq: opt E2} and combining it with \eqref{eq: opt E1} concludes the lemma.
\end{proof}

We state an important identity in the following lemma.
\begin{lemma}\label{lem: u - ubar}
	Let $\bu_1, \ldots \bu_n \in \R^d$ be $n$ arbitrary vectors. Define $\bar \bu = (\sum_{i=1}^n \bu_i)/n$. Then,
	\begin{align*}
	\sum_{i=1}^n \Vert \bu_i - \bar \bu \Vert^2 = \sum_{i=1}^n \Vert \bu_i \Vert^2 - n \Vert \bar \bu \Vert^2.
	\end{align*}
\end{lemma}
\begin{proof}
	We have
	\begin{align*}
	\sum_{i=1}^n \Vert \bu_i - \bar \bu \Vert^2 &= \sum_{i=1}^n \Vert \bu_i \Vert^2 + n \Vert \bar \bu \Vert^2 - 2\sum_{i=1}^n \langle \bu_i, \bar \bu \rangle \\
	&= \sum_{i=1}^n \Vert \bu_i \Vert^2 + n \Vert \bar \bu \Vert^2 - 2n \langle \bar \bu, \bar \bu \rangle \\
	& = \sum_{i=1}^n \Vert \bu_i \Vert^2 - n \Vert \bar \bu \Vert^2.
	\end{align*}
\end{proof}

\begin{proof}[Proof of Lemma \ref{lem: consensus 1}]
	We have,
	\begin{align}\label{eq: consensus var}
	  \sum_{i=1}^{N} \E \left[\left\Vert \bx_i^{t+1} - \bbx^{t+1} \right\Vert^2 \right] = \sum_{i=1}^{N} \left\Vert \E \left[\bx_i^{t+1} - \bbx^{t+1} \right] \right\Vert^2 + \sum_{i=1}^{N} \E \left[ \left\Vert \bx_i^{t+1} - \bbx^{t+1} - \E \left[\bx_i^{t+1} - \bbx^{t+1} \right] \right\Vert^2 \right].
	\end{align}
	Let us consider the first term on the right hand side of \eqref{eq: consensus var}. Taking conditional expectation of both sides of \eqref{eq: average x update} implies,
	\begin{align} \label{eq: consensus 1}
	\sum_{i=1}^{N} \left\Vert \E \left[\bx_i^{t+1} - \bbx^{t+1} | \ \F^{t} \right] \right\Vert^2 &= \sum_{i =1}^{N} \Vert \bx_i^t - \bbx^t - \eta_t(\bg_i^t - \bar \bg^t) \Vert^2 
	\nonumber \\
	&= \sum_{i=1}^{N} \left(\Vert \bx_i^t - \bbx^t \Vert^2 + \eta_t^2\Vert \bg_i^t - \bar \bg^t \Vert^2 - 2 \eta_t \langle \bg_i^t, \bx_i^t - \bbx^t \rangle \right).
	\end{align}
	By $L$-smoothness of $F$, $\Vert \nabla f(\bx) - \nabla f(\by) \Vert^2 \leq 2L(f(\bx) - f(\by) - \langle \nabla f(\by), \bx - \by \rangle)$. Thus,
	\begin{multline} \label{eq: consensus L}
	\sum_{i=1}^N \Vert \bg_i^t - \bar \bg^t \Vert^2 \leq 
	\sum_{i=1}^N \Vert \bg_i^t - \nabla f(\bbx^t) \Vert^2 \leq \\
	\sum_{i=1}^N 2L\left( f(\bbx^t) - f(\bx_i^t) - \langle \bg_i^t, \bbx^t - \bx_i^t \rangle \right)  
	\leq 2L\sum_{i=1}^N \langle \bg_i^t, \bx_i^t  - \bbx^t \rangle.
	\end{multline}
	Moreover, by $\mu$-strong convexity of $F$,
	\begin{align} \label{eq: consensus mu}
	\sum_{i=1}^N \langle \bg_i^t, \bx_i^t - \bbx^t \rangle \geq 	\sum_{i=1}^N \left( f(\bx_i^t) - f(\bbx_i^t) + \frac{\mu}{2}\Vert \bx_i^t - \bbx^t\Vert^2 \right) \geq \frac{\mu}{2} \sum_{i=1}^N \Vert \bx_i^t - \bbx^t\Vert^2. 
	\end{align}
	We used the Jensen's inequality $\sum_{i =1}^N f(\bx_i^t) - f(\bbx^t) \leq 0$ in both equations above.
	Combining \eqref{eq: consensus 1}-\eqref{eq: consensus mu} and having $\eta_t<1/L$ we obtain,
	\begin{align*}
	\sum_{i=1 }^{N} \Vert \E[\bx_i^{t+1} - \bbx^{t+1} | \F^t] \Vert^2 &\leq \sum_{i=1}^{N} \Vert \bx_i^t - \bbx^t \Vert^2 - (2\eta_t - 2\eta_t^2 L) \sum_{i=1}^N \langle \bg_i^t, \bx_i^t - \bbx^t \rangle \\
	& \leq \sum_{i=1}^{N} \Vert \bx_i^t - \bbx^t \Vert^2 \left( 1 - \eta_t \mu + \eta_t^2 \mu L \right).
	\end{align*}
	Now, consider the second term on the right hand side of \eqref{eq: consensus var}.
	We have,
	\begin{align*}
	\sum_{i=1}^{N} \E \left[\left\Vert \bx_i^{t+1} - \bbx^{t+1} - \E[\bx_i^{t+1} - \bbx^{t+1}] \right\Vert^2 | \F^t \right] &= 
	\sum_{i=1}^{N} \E \left[\left\Vert \bx_i^{t+1} - \E[\bx_i^{t+1}] - (\bbx^{t+1} - \E[\bbx^{t+1}]) \right\Vert^2 | \F^t \right] \\
	&= \eta_t^2 \sum_{i=1}^{N} \E \left[ \left\Vert \bw_i^t - \bbw^t \right\Vert^2 | \F^t \right] \\
	&= \eta_t^2 \left(\sum_{i=1}^{N} \E \left[ \left\Vert \bw_i^t \right\Vert^2 | \F^t \right]- N \E \left[ \left\Vert \bbw^t \right\Vert^2 |\F^t \right]\right) \\
	&= \eta_t^2 \sum_{i=1}^{N} \E \left[ \left\Vert \bw_i^t \right\Vert^2 | \F^t \right](1-\frac{1}{N}) \\
	&\leq (N-1)\eta_t^2\sigma^2 + (1-\frac{1}{N})\eta_t^2 c\sum_{i=1}^N \Vert \bg_i^t \Vert^2,
	\end{align*}
	where $\bw_i^t$ are defined at the beginning of this section and $\bbw^t := (\sum_{i=1}^N \bw_i^t)/n$ and we used Lemma \ref{lem: u - ubar} in the third equation and the conditional independence of $\bw_i^t$ to use $\E[\Vert \bbw^t \Vert^2 | \F^t] = (1/N^2)\sum_{i=1}^N \E[\Vert \bw_i^t \Vert^2 | \F^t]$ in the last equality.
	Taking full expectation of the two relations above with respect to $\F^t$ and combining them with \eqref{eq: consensus var} completes the proof.
\end{proof}

Before proving Lemma \ref{lem: consensus 2}, let us state and prove the following lemma.
\begin{lemma}\label{lem: products}
	Let $b \geq a > 2$ be integers. Define $\Phi(a,b) = \prod_{i=a}^b \left( 1 - \frac{2}{i} \right)$. We then have $
	\Phi(a,b) \leq \left( \frac{a}{b+1} \right)^{2}.$
\end{lemma}
\begin{proof}
	Indeed,
	\begin{align*}
	\ln (\Phi(a,b))  =  \sum_{i=a}^b \ln \left( 1 - \frac{2}{i} \right)   
	\leq  \sum_{i=a}^b - \frac{2}{i} 
	\leq  -  2\left[ \ln (b+1) - \ln (a) \right].
	\end{align*}
	where we used the inequality $\ln (1-x) \leq -x$ as well as the standard technique of viewing $\sum_{i=a}^b 1/i$ as a Riemann sum for $\int_{a}^{b+1} 1/x ~dx$ and observing that the Riemann sum overstates the integral. Exponentiating both sides now implies the lemma.
\end{proof}

\begin{proof}[Proof of Lemma \ref{lem: consensus 2}]
	Define $a^k =  \E\left[ \sum_{i=1}^N \Vert \bx_i^{k} - \bbx^k \Vert^2 \right]$ and $\Delta_k = (1 - \eta_k \mu + \eta_k^2 \mu L)$ for $k\geq 0$ . By Lemma \ref{lem: consensus 1},
	\begin{align*}
	a^{t} &\leq \Delta_{t-1}a^{t-1}  + \eta_{t-1}^2(N-1)(\sigma^2 + c \E[G^{t-1}]) \\
	& \leq \Delta_{t-1}(\Delta_{t-2}a^{t-2} + \eta_{t-2}^2(N-1)(\sigma^2+c\E[G^{t-2}])) + \eta_{t-1}^2(N-1)(\sigma^2+c\E[G^{t-1}]) \\
	& \leq \ldots \leq \prod_{k=\tau(t)}^{t-1} \Delta_k a^{\tau(t)} + 
	(N-1)\sum_{k=\tau(t)}^{t-1} \eta_k^2(\sigma^2+c\E[G^k])\prod_{i=k+1}^{t-1} \Delta_i \\
	&= (N-1)\sum_{k=\tau(t)}^{t-1} \eta_k^2(\sigma^2+c\E[G^k])\prod_{i=k+1}^{t-1} \Delta_i,
	\end{align*}
	where we used $a^{\tau(t)}= 0$ in the last equation. 
	By the choice of stepsize and $\beta\geq 9\kappa $,  we have 
	\begin{align*}
	\Delta_k = 1-\frac{3}{k+\beta} + \frac{9L}{\mu (k+\beta)^2} \leq 1 - \frac{3}{k+\beta} + \frac{9\kappa}{(k+\beta)\beta}
	\leq 1 - \frac{3}{k+\beta} + \frac{1}{(k+\beta)} =  1 - \frac{2}{k+\beta}.
	\end{align*}
	Therefore, by Lemma \ref{lem: products}, 
	\begin{align*}
	a^t \leq (N-1) \sum_{k = \tau(t)}^{t-1} \frac{9(\sigma^2+c\E[G^k])}{\mu^2 (k+\beta)^2 } \frac{(k+\beta+1)^2}{(t + \beta)^2} \leq (N-1) \sum_{k = \tau(t)}^{t-1} \frac{12(\sigma^2+c\E[G^k])}{\mu^2 (t+\beta)^2},
	\end{align*}
	where we used $9(k+\beta + 1)^2/(k+\beta)^2\leq 9(\beta+1)^2/\beta^2 \leq 9(10/9)^2\leq 12$ since $\beta \geq 9\kappa \geq 9$.
\end{proof}

\begin{proof}[Proof of Lemma \ref{lem: variance}]
	We have,
	\begin{align*}
	\E[\Vert \tbg^t \Vert^2|\F^t] = \E[\Vert \bbg^t + \bar \e^t \Vert^2 | \F_t] = \Vert \bbg^t \Vert^2 + \E[\Vert \bbw^t \Vert^2| \F^t] \leq  \frac{1}{N}\sum_{i=1}^N \Vert \bg_i^t \Vert^2 + \frac{1}{N^2}\sum_{i=1}^N(\sigma^2 + c\Vert \bg_i^t \Vert^2),
	\end{align*}
	where in the last inequality we used Lemma \ref{lem: u - ubar} and the conditional independency of $\bw_i^t$ to decouple the noise terms.
\end{proof}

\begin{proof}[Proof of Theorem \ref{thm: general}]
	Combining Equations Lemmas \ref{lem: error decay}-\ref{lem: variance} and plugging $\eta_t = 3/(\mu(t+\beta))$ we obtain
	\begin{multline*}
	\xi^{t+1} \leq \xi^t(1 - \mu \eta_t) + 
	\frac{18L^2}{\mu^3(t+\beta)^3} 
	\sum_{k=\tau(t)}^{t-1} \left(c\E[G^k] + \sigma^2 \right) \\
	+ \frac{9L}{2\mu^2(t+\beta)^2} \left( \left(1+\frac{c}{N} \right) \E[G^t] + \frac{\sigma^2}{N} \right)
	- \frac{3}{2\mu(t+\beta)} \E[G^t].
	\end{multline*}
	Let us multiply both sides of relation above by $(t+\beta)^2$ 
	and use the following inequality
	\begin{align*}
	(1-\mu \eta_t)(t+\beta)^2 = \left(1 - \frac{2}{t+\beta} \right)(t+\beta)^2 = (t+\beta)^2 - 2(t+\beta) < (t+\beta-1)^2,
	\end{align*}
	to obtain,
	\begin{multline*}
	\xi^{t+1} (t+\beta)^2 \leq \xi^t (t+\beta-1)^2 + \frac{9L \sigma^2}{2 \mu^2 N} + \\
	\frac{18L^2}{\mu^3(t+\beta)} \sum_{k=\tau(t)}^{t-1} \left( c\E[G^k]+\sigma^2 \right)  + 
	\left( \frac{9L}{2\mu^2} \left(1+\frac{c}{N} \right) - \frac{3(t+\beta)}{2\mu} \right)\E[G^t].
	\end{multline*}
	Summing relation above for $t=\tau_i,\ldots,\tau_{i+1}-1$, where $\tau_i,\tau_{i+1}\in \I$ are two consecutive communication times, implies,
	\begin{multline*}
	\xi^{\tau_{i+1}}(\tau_{i+1} + \beta -1 )^2 \leq \xi^{\tau_i}(\tau_{i} + \beta -1 )^2  +  \frac{9L \sigma^2}{2 \mu^2 N}(\tau_{i+1} - \tau_i)
	+ \frac{18L^2\sigma^2}{\mu^3}\sum_{t=\tau_i}^{\tau_{i+1}-1} \frac{t-\tau_i}{t+\beta} \\
	+ \sum_{t=\tau_i}^{\tau_{i+1}-1}\E[G^t] \left( \sum_{k=t+1}^{\tau_{i+1}-1}\frac{18 L^2 c}{\mu^3(k+\beta)} + \frac{9L}{2\mu^2} \left(1+\frac{c}{N} \right) - \frac{3 (t+\beta)}{2 \mu} \right).
	\end{multline*}
	Each of the coefficients of $\E[G^t]$ in above can be bounded by,
	\begin{scriptsize}
		\begin{align*}
		\sum_{k=t+1}^{\tau_{i+1}-1}\frac{18L^2 c}{\mu^3(k+\beta)} + \frac{9L}{2\mu^2} \left(1+\frac{c}{N} \right) - \frac{3(t+\beta)}{2\mu} &\leq 
		\frac{18L^2c}{\mu^3} \ln \left(\frac{\tau_{i+1}+\beta -1 }{\tau_i + \beta} \right)
		+ \frac{9L}{2\mu^2} \left(1+\frac{c}{N} \right) - \frac{3\tau_i +\beta}{2\mu} \\
		&= \frac{3}{2\mu}\left( 12\kappa^2c \ln \left(1 + \frac{H_i - 1}{\tau_i + \beta} \right) + 3\kappa \left(1+\frac{c}{N} \right) - (\tau_i + \beta)\right) \\
		&\leq 0,
		\end{align*}
	\end{scriptsize}
	where we used $\sum_{k=t_1+1}^{t_2} 1/k \leq \int_{t_1}^{t_2} dx/x = \ln(t_2/t_1)$ in the first inequality and the last inequality comes from the assumption of the theorem. Now that the coefficients of $\E[G^k]$ are non-positive, we can simply ignore them and obtain,
	\begin{align*}
	\xi^{\tau_{i+1}}(\tau_{i+1} + \beta -1 )^2 \leq \xi^{\tau_i}(\tau_{i} + \beta -1 )^2  +  \frac{9L \sigma^2}{2 \mu^2 N}(\tau_{i+1} - \tau_i)
	+ \frac{18L^2\sigma^2}{\mu^3}\sum_{t=\tau_i}^{\tau_{i+1}-1} \frac{t-\tau_i}{t+\beta}.
	\end{align*}
	Recursing relation above for $i=0,\ldots,R-1$ implies,
	\begin{align*}
	\xi^{T}(T + \beta -1 )^2 \leq \xi^{0}( \beta -1 )^2  +  \frac{9L \sigma^2}{2 \mu^2 N}T
	+ \frac{18L^2\sigma^2}{\mu^3}\sum_{t=0}^{T-1} \frac{t-\tau(t)}{t+\beta}.
	\end{align*}
	Dividing both sides by $(T+\beta -1 )^2$ concludes the proof.
\end{proof}

\begin{proof}[Proof of Theorem \ref{thm: logT}]
	We have,
	\begin{align*}
	\tau_j = \tau_0 + \sum_{i=0}^{j-1} H_i = a \frac{j(j+1)}{2}, \qquad j=0,\ldots,k-1.
	\end{align*}
	Hence,
	\begin{align*}
	1 + \frac{H_0-1}{\tau_0 + \beta} &= 1 + \frac{a-1}{\beta} \leq 1 + \frac{2T}{9 \kappa R^2} \leq 1 + \frac{T}{4 \kappa R^2},\\
	1 + \frac{H_i-1}{\tau_i + \beta} &\leq 1 + \frac{a(i+1)}{\frac{ai(i+1)}{2}} \leq 3, \qquad i\geq 1.
	\end{align*}
	Thus, $12 \kappa^2c \ln(1 + \frac{H_i - 1}{\tau_i + \beta}) + 3 \kappa(1+\frac{c}{N}) - (\tau_i + \beta)\leq 0, i=0,\ldots,R-1$ and we can use Theorem \ref{thm: general}.
	Moreover,
	\begin{align*}
	\sum_{t=0}^{T-1} \frac{t-\tau(t)}{t+\beta} &\leq \sum_{j=0}^{R-1} \sum_{i=1}^{H_j - 1} \frac{i}{\tau_j + i + \beta } \leq H_0 + \sum_{j=1}^{R-1} \sum_{i=1}^{H_j - 1} \frac{i}{\tau_j + 1 + \beta } \\
	&= a + \sum_{j=1}^{R-1} \frac{H_j (H_j - 1)}{2(\tau_j + 1 + \beta)} = a + \sum_{j=1}^{R-1} \frac{a(j+1) (a(j+1) - 1)}{a j(j+1) + 2(1 + \beta)} \\
	&\leq a + \sum_{j=1}^{R-1} \frac{a^2 (j+1)^2 }{aj(j+1)} \leq 2a R.
	\end{align*}
	Plugging the values of $R$ and $a$ implies,
	\[ \sum_{t=0}^{T-1} \frac{t-\tau(t)}{t+\beta} \leq 2aR \leq 2(\frac{2T}{R^2}+1)R = \frac{4T}{R} + 2R \leq \frac{4T}{R} + \frac{4T}{R} = \frac{8T}{R}, \]
	where we used $R\leq \sqrt{2T}$ in the last inequality. Using the relation above together with Theorem \ref{thm: general} concludes the proof.
\end{proof}
\newpage
\section{One-shot averaging}\label{sec: apx osa}
In this section we prove Theorem~\ref{thm: osa main} for one-shot averaging.
The main idea is to use second order approximation for gradients at any point with respect to the minimizer and show that the residual errors are \emph{insignificant}, using concentration results from \citet{karimi2016linear}.

\subsection{Preliminaries}
Define $\bv(\by,\bx) :=  \nabla f(\by) - \left(\nabla f(\bx) + \nabla^2 f(\bx)(\by-\bx) \right)$ and $\bv_i^t = \bv(\bx_i^t, \bx^*)$.

\begin{lemma}\label{lem: linear estimate}
	Let Assumption \ref{asm: 2-time diff} hold. Then $| [\bv(\bx, \bx^*)]_i | = o(\Vert \bx - \bx^* \Vert)$ for $i=1,\ldots,d$.
\end{lemma}
\begin{proof}
	Denote $h_i(\bx) = [\nabla f(\bx)]_i$. Then by Assumption \ref{asm: 2-time diff}, $h_i$ is continuously differentiable over an open set containing $\bx^*$. Thus,
	\begin{align*}
	h_i(\bx) &= h_i(\bx^*) + \nabla h_i(\bx^*)^\top (\bx - \bx^*) + o(\Vert \bx - \bx^* \Vert) 
	= \nabla h_i(\bx^*)^\top (\bx - \bx^*) + o(\Vert \bx - \bx^* \Vert).
	\end{align*}
	Therefore,
	\begin{align*}
	[\bv(\bx, \bx^*)]_i &= h_i(\bx) -  \sum_{j=1}^{d} \frac{\partial^2f}{\partial x_i \partial x_j}(\bx^*) [\bx - \bx^*]_j 
	= h_i(\bx) - \nabla h_i(\bx^*)^\top (\bx - \bx^*) = o(\Vert \bx - \bx^* \Vert).
	\end{align*}
\end{proof}

Let us define $u(r) := \max_{\Vert \bx-\bx^* \Vert \leq r} \Vert\bv(\bx,\bx^*)\Vert$. We have $u(r)=o(r)$.

\begin{theorem}[\cite{karimi2016linear}, Theorem 1] \label{thm: quadratic growth (QG)}
	Under Assumptions \ref{asm: smoothness} and \ref{asm: PL}, the following inequality, known as the quadratic growth (QG) condition holds:
	\begin{align*}
	\Vert \bx - \bx^*\Vert^2 \leq \frac{2}{\mu} (f(\bx) - f^*).
	\end{align*}
\end{theorem}

\begin{lemma}\label{lem: A> mu}
	Under Assumptions \ref{asm: smoothness}, \ref{asm: PL} and \ref{asm: 2-time diff} we have,
	\[  \nabla^2 f(\bx^*)  \succeq \mu.\]
\end{lemma}
\begin{proof}
	The result is established by using the \emph{linear approximation theorem} on a sequence of points converging to $x^*$ on a line, continuity of Hessian as well as the quadratic growth from Theorem \ref{thm: quadratic growth (QG)}. Similar approach can be found in the proof of Theorem 2.26 \cite{beck2014introduction}.
\end{proof}

Next, we state a Theorem from \cite{madden2020high} which we will use frequently in the rest of our results. 
\begin{theorem}[\cite{madden2020high}, Theorem 4 and 13] \label{thm: concentration}
	Under Assumptions \ref{asm: smoothness}, \ref{asm: PL} and \ref{asm: noise sub gaussian},
	SGD with step-size sequence $\{\eta_t\} = \{\theta_t\}$ defined in \eqref{eq: theta}, constructs a sequence of $\{\bx^t\}$ such that there exist $C_1,C_2>0$ such that for $t\geq t_0$,
	\begin{align*}
	\E[f(\bx^t)] - f^* = C_1 \frac{L \sigma^2}{\mu^2 t},
	\end{align*}
	and w.p. $\geq 1-\delta$ for all $\delta \in (0,1/e)$,
	\begin{align*}
	f(\bx^t) - f^* \leq C_2 \frac{L\sigma^2 \log(e/\delta)}{\mu^2 t}.
	\end{align*}
\end{theorem}

\begin{lemma}
	Under Assumptions \ref{asm: smoothness} and \ref{asm: 2-time diff} we have,
	\begin{align}\label{eq: v linear growth}
		\Vert \bv(\bx, \bx^*) \Vert \leq 2L\Vert \bx - \bx^*\Vert.
	\end{align}
\end{lemma}
\begin{proof}
	We have,
	\begin{align*}
	\Vert \bv(\bx, \bx^*) \Vert &= \Vert \nabla f(\bx) - \nabla^2 f(\bx^*) (\bx - \bx^*) \Vert  \\
	&\leq \Vert \nabla f(\bx) \Vert + \Vert \nabla^2 f(\bx^*)(\bx - \bx^*) \Vert  \\
	&\leq L\Vert \bx - \bx^* \Vert + \Vert \nabla^2 f(\bx^*) \Vert_2 \Vert \bx - \bx^*\Vert  \\
	&\leq 2L\Vert \bx - \bx^*\Vert, 
	\end{align*}
	where we used $\Vert \nabla^2 f(\bx^*) \Vert \leq L$ in the last inequality.
\end{proof}

The following lemma is the key result we need to show the asymptotic performance of OSA.
\begin{lemma}\label{lem: vt o(1/t)}
	Under Assumptions \ref{asm: smoothness}, \ref{asm: PL}, \ref{asm: 2-time diff} and \ref{asm: noise sub gaussian} and steps-size sequence $\{ \eta_t\} = \{ \theta_t\}$ defined in \eqref{eq: theta}, we have
	\begin{enumerate}
		\item $ \E[\Vert \bv_i^t \Vert^2] = o(\frac{1}{t})$,
		\item $ \E [\Vert \bv_i^t \Vert \Vert \bx_i^t - \bx^* \Vert] = o(\frac{1}{t})$.
	\end{enumerate}
\end{lemma}
\begin{proof}
	Let us define $u(r) := \max_{\Vert \bx-\bx^* \Vert \leq r} \Vert\bv(\bx,\bx^*)\Vert$. By Lemma \ref{lem: linear estimate} we have $u(r)=o(r)$. 
	Also define random variable $r_i^t = \Vert \bx_i^t - \bx^* \Vert$. 
	
	Since $u(r) = o(r)$, for any $\epsilon>0$ there exists $s>0$ such that for $r \leq s$, $u(r) \leq \sqrt{\epsilon} r$ or $u(r)^2 \leq \epsilon r^2$. We have,
	\begin{align}
	\E[\Vert \bv_i^t \Vert^2] &= \E_{\bx_i^t}[\Vert \bv(\bx_i^t, \bx^*) \Vert^2] \leq \E_{r_i^t} [u(r_i^t)^2] \nonumber \\
	&= \int_{0}^{\infty} u(r)^2 p_{r_i^t}(r) dr \nonumber \\
	&= \int_{0}^s u(r)^2 p_{r_i^t}(r)dr + \int_{s}^{\infty} u(r)^2 p_{r_i^t}(r)dr  \nonumber\\
	&\leq \epsilon\int_{0}^s r^2 p_{r_i^t}(r)dr + 4L^2 \int_{s}^{\infty} r^2 p_{r_i^t}(r)dr \nonumber \\
	&= \epsilon \E[(r_i^t)^2] + (4L^2-\epsilon)\int_{s}^{\infty} r^2 p_{r_i^t}(r)dr, \label{eq: v2 bound1}
	\end{align}
	where $p_X$ denotes the Probability Density Function (PDF) for random variable $X$ and we used $u(r) \leq 2Lr$ from \eqref{eq: v linear growth}.
	
	Without loss of generality, we assume $t\geq t_0$ for the rest of the proof.
	By Theorems \ref{thm: quadratic growth (QG)} and \ref{thm: concentration} we have,
	\[ \E \left[(r_i^t)^2 \right] = \E \left[ \left\Vert \bx_i^t - \bx^* \right\Vert^2 \right] \leq  \frac{2}{\mu} \E[f(\bx_i^t) - f^*] \leq 
	\frac{2C_1 L \sigma^2}{\mu^2 t}=\O\left( \frac{1}{t} \right). \]
	Moreover, define $J_t(\delta) := C_2 L \sigma^2 \log (e/\delta)/(\mu^2 t)$. Then,
	\begin{align}\label{eq: F < 1 - delta}
	\Pr \left((r_i^t)^2 \leq \frac{2J_t(\delta)}{\mu} \right) \geq \Pr \left(f(\bx_i^t) - f^* \leq J_t(\delta) \right) \geq 1-\delta, \qquad \text{for }  \delta \in (0,1/e),
	\end{align}
	or, 
	\begin{align}\label{eq: F and J}
	F_{(r_i^t)^2}^{-1}(1-\delta) \leq \frac{2J_t(\delta)}{\mu}, \qquad \text{for }  \delta \in (0,1/e),
	\end{align}
	where $F_X$ denotes the Cumulative Distribution Function (CDF) for random variable $X$.
	Since $\lim_{t \rightarrow \infty} J_t(\delta) = 0$, $\exists t_1 \geq t_0$ such that for $t\geq t_1$, $J_t^{-1}(\mu s^2/2) \in (0,1/e)$. It follows that,
	\begin{align}
	{\int_{s}^{\infty} r^2 p_{r_i^t}(r)dr} & = 	{\int_{s^2}^{\infty} r^2 p_{(r_i^t)^2}(r^2)dr^2} 
	= {\int_{s^2}^{\infty} r^2 dF_{(r_i^t)^2}(r^2)} \nonumber \\
	&= \tcb{\int_{F_{(r_i^t)^2}(s^2)}^{1} F_{(r_i^t)^2}^{-1}(x) dx} 
	= \tcb{\int_{1-F_{(r_i^t)^2}(s^2)}^{0} - F_{(r_i^t)^2}^{-1}(1-\delta) d\delta} \nonumber \\
	&=\tcb{\int_{0}^{1-F_{(r_i^t)^2}(s^2)} F_{(r_i^t)^2}^{-1}(1-\delta) d\delta} \nonumber \\
	&\leq \frac{2}{\mu} \int_{0}^{1-F_{(r_i^t)^2}(s^2)} J_t(\delta)d\delta \leq \tcr{\frac{2}{\mu} \int_{0}^{J_t^{-1}(\frac{\mu s^2}{2})} J_t(\delta)d\delta}. \label{eq: integral r^2 bound}
	\end{align}
	In the equation above, we switched from Probability Density Function (PDF) $p_{r_i^t}$ to $p_{(r_i^t)^2}$ in the first equality. In the next equality we used $p_X = d F_X/ dX$ that holds for any continuous random variable $X$. In third equality, we simply changed variable to $x=F_{(r_i^t)^2(r^2)}$ and without loss of generality we define $F_X^{-1}(y) := \inf \{x | F_X(x) \geq y \}$. In the next equation, again, we simply changed variable to $\delta = 1 - x$. Finally, in the last two inequalities, we used \eqref{eq: F and J} and
	a direct result of \eqref{eq: F < 1 - delta}, $1-F_{(r_i^t)^2}(s^2) \leq J_t^{-1}(\frac{\mu s^2}{2})$ (see Figure \ref{fig: integrals}). 
	
		\begin{figure}
	
		\centering
		\begin{tikzpicture}
		\begin{axis}[
		ymin = 0.5, ymax=1,
		xmin = 2, xmax=10,
		xlabel = {$r^2$}, 
		ylabel = {$1-\delta$},
		legend pos=south east,
		xtick = {2,6,10},
		extra x ticks = {4,5},
		extra x tick labels={$\frac{2J(1/e)}{\mu}$,$s^2$},
		ytick = {0.5,1},
		extra y ticks = {1-exp(-1) ,0.8895, 0.7769},
		extra y tick labels={$1 - 1/e$, $F_{(r_i^t)^2}(s^2)$, $1 - J_t^{-1}(\frac{\mu s^2}{2})$}
		]
		\addplot[
		color=red,
		domain=4:10,
		thick,
		]
		{max(1 - exp(1 - x/2),1 - exp(-1))};
		\addlegendentry{$1 - J_t^{-1}(\frac{\mu r^2}{2})$}
		
		\addplot[
		color=blue,
		domain=2:10,
		thick,
		]
		{max(1 - exp(0.5-x/1.85),min(x*x/12,0.6))};
		\addlegendentry{$F_{(r_i^t)^2}(r^2)$}
		
		\addplot[
		xbar interval,
		color=red,
		bar width=0.02,
		domain=5:10,
		fill=red, opacity=0.2,
		samples=20
		]
		{1 - exp(1-x/2)};
		
		\addplot[
		xbar interval,
		color=blue,
		bar width=0.015,
		domain=5:10,
		fill=blue, opacity=0.5,
		samples=15
		]
		{1 - exp(0.5-x/1.85)};
		
		\addplot[
		color=black,
		thick,
		dotted,
		]
		coordinates{(5,0.8895) (5,0.7769)(5,0.4)};
		
		\addplot[
		color=black,
		thick,
		dotted,
		]
		coordinates{(2,0.6321) (4,0.6321)(4,0.5)};
		
		
		\end{axis}
		\end{tikzpicture}
		\caption{Illustration of integrals in \eqref{eq: integral r^2 bound}}
		\label{fig: integrals}
	\end{figure}
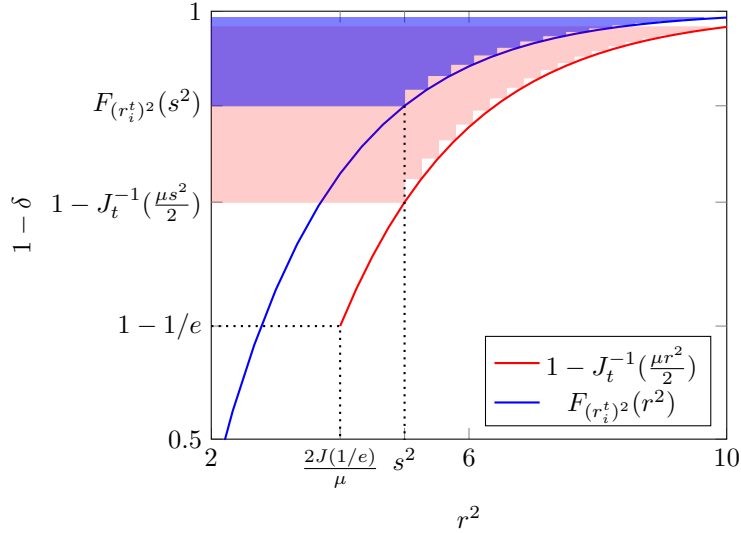
	
	By Lemma \ref{lem: int o(t)}, $\exists t_2 \geq t_1$ such that for $t\geq t_2$,
	$ \int_{0}^{J_t^{-1}(\mu s^2 / 2)} J_t(\delta)d\delta \leq \epsilon B_1/t$, where  $B_1 := 2C_2 L \sigma^2 e / \mu^2$.
	Combining with \eqref{eq: v2 bound1} we obtain,
	\begin{align*}
	\E[\Vert \bv_i^t \Vert^2 ] \leq  \frac{\epsilon}{t}\left( \frac{2C_1 L \sigma^2}{\mu^2} + \frac{16C_2L^3\sigma^2 e}{\mu^3}  \right), \qquad t \geq t_2.
	\end{align*}
	
	Next, we show $\E[\Vert \bv_i^t \Vert \Vert \bx_i^t - \bx^* \Vert] = o(1/t)$.
	Since $u(r) = o(r)$, for any $\epsilon>0$, there exists $s'>0$ such that for $r \leq s'$, $u(r) \leq \epsilon r$. Then,
	\begin{align*}
	\E[\Vert \bv_i^t \Vert \Vert \bx_i^t - \bx^* \Vert] &\leq \E[u(r_i^t) r_i^t] \\
	&= \int_{0}^{s'} u(r)r p_{r_i^t}(r) dr + \int_{s'}^{\infty} u(r)r p_{r_i^t}(r) dr \\
	& \leq \epsilon \int_{0}^{s'}  r^2 p_{r_i^t}(r) + 4L^2 \int_{s'}^{\infty} r^2 p_{r_i^t}(r)dr.
	\end{align*}
	Following the same steps from the first part of this proof, we obtain $\exists t_3>0$ such that,
	\begin{align*}
	\E[\Vert \bv_i^t \Vert \Vert \bx_i^t - \bx^* \Vert] \leq \frac{\epsilon}{t}\left( \frac{2C_1 L \sigma^2}{\mu^2} + \frac{16C_2L^3\sigma^2 e}{\mu^3}  \right), \qquad t \geq t_3.
	\end{align*}
	Since we could pick $\epsilon$ arbitrarily small, we showed that $\E[\Vert \bv_i^t \Vert^2 ] = o(1/t)$ and $\E[\Vert \bv_i^t \Vert \Vert \bx_i^t - \bx^* \Vert] = o(1/t)$.
\end{proof}

\begin{lemma}\label{lem: int o(t)}
	Let $q_t:(0,1/e) \rightarrow \R_+$ be defined as
	$q_t(\delta) = a_1 \log(e / \delta)/t $
	for some $a_1>0$ and $\forall t\geq 1$.
	Suppose $y \in range(q_t)$ for $t\geq t_1$, then for any $\epsilon>0$, there exists $t_2\geq t_1$ such that for any $t\geq t_2$,
	$$ \int_{0}^{q_t^{-1}(y)} q_t(\delta)d\delta \leq \frac{B \epsilon}{t},$$
	where $B = 2a_1 e$.
\end{lemma}
\begin{proof}
	Define $x_t$ such that $q_t(x_t) = y$. Then,
	\begin{align}\label{eq: x to y}
	\frac{a_1 \log(e /x_t)}{t} = y \iff \log(\frac{e}{x_t})= \frac{y t}{a_1}	
	\iff x_t =  \exp(1 - \frac{yt}{a_1}).
	\end{align}
	
	Moreover, 
	\begin{align*}
	\int_{0}^{x_t} q_t(\delta)d\delta &= \frac{a_1}{t} \int_{0}^{x_t} \log(\frac{e}{\delta})d\delta \\
	&=\frac{a_1}{t}\left(x_t - x_t (\log(x_t) - 1)\right) \\
	&=\frac{a_1}{t}\left(x_t +  x_t (\frac{y t}{a_1})\right) 
	=x_t(y + \frac{a_1}{t}).,
	\end{align*}
	where we used \eqref{eq: x to y} in third equality.
	First, we note that for $t \geq a_1 /y$, we have $y + a_1 /t \leq 2y$.
	Next, we show that for $t$ large enough, $x_t \leq B\epsilon / (2yt)$ for some $B>0$. We have $\lim_{s\rightarrow \infty}\exp(s)/s = \infty$. Therefore $\exists s_0 \geq 1$ such that for $s\geq s_0$, $\exp(s)/s \geq 1/\epsilon$. Thus for $t \geq s_0 a_1 /y$ we have,
	\begin{align*}
	&\exp(\frac{ty}{a_1 }) \geq \frac{ty}{a_1 \epsilon},\\
	\Rightarrow & x_t =  \exp(1 - \frac{yt}{a_1}) \leq e (\frac{a_1 \epsilon}{t y }) = \frac{B \epsilon}{2yt},
	\end{align*}
	where $B := 2 a_1 e $.
	Therefore, for $t \geq t_2:=\max\{s_0 a_1 /y, t_1\}$ we have $	\int_{0}^{x_t} q_t(\delta)d\delta \leq 2x_ty \leq B\epsilon/t$.
\end{proof}

\subsection{One-step progress}
\begin{lemma}\label{lem: one step}
	Under Assumptions \ref{asm: smoothness}, \ref{asm: PL}, \ref{asm: 2-time diff} and \ref{asm: noise sub gaussian} and steps-size sequence $\{ \eta_t\} = \{ \theta_t\}$ defined in \eqref{eq: theta}, we have
	\begin{align}\label{eq: one-step progress}
	\E[\Vert \bbx^{t+1} - \bx^* \Vert^2] \leq (1-\eta_t \mu)^2 \E[\Vert \bbx^t - \bx^* \Vert^2] + \frac{\eta_t^2 \sigma^2}{N} + o\left(\frac{1}{t^2}\right).
	\end{align}
\end{lemma}
\begin{proof}
	Let us define $\bA = \nabla^2 f(\bx^*)$. By definition,
	\begin{align}\label{eq: vi}
	\nabla f(\bx_i^t) = \bA(\bx _i^t- \bx^*) + \bv_i^t.
	\end{align}
	Plugging \eqref{eq: vi} in SGD process and averaging over all $i$ we obtain,
	\begin{align*}
	\bbx^{t+1} &= \bbx^t - \frac{\eta_t}{N}\sum_{i=1}^N \hbg_i^t 
	= \bbx^t - \frac{\eta_t}{N}\sum_{i=1}^N \left( \nabla f(\bx_i^t) + \bw_i^t \right) \\
	&= \bbx^t - \frac{\eta_t}{N}\sum_{i=1}^N \left( \bA(\bx_i^t - \bx^*) + \bv_i^t+  \bw_i^t \right) 
	= \bbx^t - \eta_t \bA (\bbx^t - \bx^*) + \frac{\eta_t}{N} \sum_{i=1}^N (\bv_i^t + \bw_i^t).
	\end{align*}
	Thus,
	\begin{align*}
	\E[\Vert \bbx^{t+1} - \bx^* \Vert^2| \F_t] &= 
	\E[\Vert (I - \eta_t \bA)(\bbx^t - \bx^*) + \frac{\eta_t}{N} \sum_{i=1}^N (\bv_i^t + \bw_i^t) \Vert^2 | \F_t]\\
	&= \E[\Vert (I - \eta_t \bA)(\bbx^t - \bx^*) + \frac{\eta_t}{N} \sum_{i=1}^N \bv_i^t \Vert^2 | \F_t] + \E[ \Vert \frac{\eta_t}{N} \sum_{i=1}^N \bw_i^t \Vert^2 | \F_t] \\
	& \leq  \E[\Vert (I - \eta_t \bA)(\bbx^t - \bx^*) + \frac{\eta_t}{N} \sum_{i=1}^N \bv_i^t \Vert^2 | \F_t] + \frac{\eta_t^2 \sigma^2}{N}.
	\end{align*}
	Taking full expectation with respect to $\F_t$ yields,
	\begin{align}
	\E[\Vert \bbx^{t+1} - \bx^* \Vert^2] &	\leq \E[\Vert (I - \eta_t \bA)(\bbx^t - \bx^*) + \frac{\eta_t}{N} \sum_{i=1}^N \bv_i^t  \Vert^2] + \frac{\eta_t^2 \sigma^2}{N} \nonumber\\
	& = \E[\Vert (I - \eta_t \bA)(\bbx^t - \bx^*)\Vert^2]  + \frac{\eta_t^2 \sigma^2}{N} \nonumber\\
	&+	\underbrace{\E[\Vert \frac{\eta_t}{N} \sum_{i=1}^N \bv_i^t \Vert^2]}_{T_2}
	+ \underbrace{\E[\Vert (I - \eta_t \bA)(\bbx^t - \bx^*)\Vert \Vert \frac{\eta_t}{N} \sum_{i=1}^N \bv_i^t \Vert]}_{T_3}. \label{eq: one-step 1}
	\end{align}
	Next we bound $T_2$ and $T_3$. Using Lemma \ref{lem: vt o(1/t)} we have,
	\begin{align}\label{eq: T2}
	\E[\Vert \frac{\eta_t}{N} \sum_{i=1}^N \bv_i^t \Vert^2] \leq \frac{\eta_t^2}{N} \sum_{i=1}^N \E[\Vert \bv_i^t \Vert^2] \leq \frac{4}{\mu^2 t^2} o\left(\frac{1}{t} \right) = o \left(\frac{1}{t^3} \right).
	\end{align}
	Moreover, by Lemma \ref{lem: A> mu} and $f$ being $L$-smooth we have  $\mu \preceq \bA \preceq L $. It follows
	\[1 - \eta_t L \preceq I - \eta_t\bA \preceq 1 - \eta_t \mu. \]
	Since $\eta_t \leq 1/L$ and $I - \eta_t\bA$ is symmetric, we have $\Vert I - \eta_t\bA \Vert \leq 1 - \eta_t \mu \leq 1$. Then,
	\begin{align*}
	\Vert (I - \eta_t \bA)(\bbx^t - \bx^*)\Vert \leq \Vert I - \eta_t \bA \Vert \Vert \bbx^t - \bx^* \Vert \leq \Vert \bbx^t - \bx^* \Vert.
	\end{align*}
	Thus,
	\begin{align} 
	\E[\Vert (I - \eta_t \bA)(\bbx^t - \bx^*)\Vert \Vert \frac{\eta_t}{N} \sum_{i=1}^N \bv_i^t \Vert] &\leq \E[\Vert \bbx^t - \bx^*\Vert \Vert \frac{\eta_t}{N} \sum_{i=1}^N \bv_i^t \Vert] \nonumber \\
	&\leq \E\left[ \left( \frac{1}{N} \sum_{i=1}^{N}\Vert \bx_i^t - \bx^* \Vert \right) \left( \frac{\eta_t}{N} \sum_{i=1}^{N}\Vert \bv_i^t \Vert \right) \right] \nonumber \\
	&\leq \frac{\eta_t}{N^2} \sum_{i=1}^{N}\E[ \Vert \bx_i^t - \bx^* \Vert \Vert \bv_i^t \Vert ] + \frac{\eta_t}{N^2} \sum_{i\neq j}\E[ \Vert \bx_j^t - \bx^* \Vert \Vert \bv_i^t \Vert ] \nonumber \\
	& = \frac{\eta_t}{N^2} \sum_{i=1}^{N}\E[ \Vert \bx_i^t - \bx^* \Vert \Vert \bv_i^t \Vert ] + \frac{\eta_t}{N^2} \sum_{i\neq j}\E[ \Vert \bx_j^t - \bx^* \Vert] \E[\Vert \bv_i^t \Vert ] \nonumber\\
	&\leq \frac{2}{\mu t} o \left(\frac{1}{t} \right ) + \frac{2}{\mu t}o\left(\frac{1}{\sqrt{t}} \right) o \left(\frac{1}{\sqrt{t}} \right) = o \left(\frac{1}{t^2} \right), \label{eq: T3}
	\end{align}
	where we used $|E[X]| \leq \sqrt{E[X^2]}$ for random variables $\Vert \bx_j^t - \bx^* \Vert$ and $\bv_i^t$ and Lemma \ref{lem: vt o(1/t)} in last equation above.
	plugging \eqref{eq: T2} and \eqref{eq: T3} in \eqref{eq: one-step 1} we obtain the desired result.
\end{proof}

Now we are ready to present the proof of Theorem \ref{thm: osa main}.
\begin{proof}[Proof of Theorem \ref{thm: osa main}]
	Denote $\psi^t := \E[\Vert \bbx^t - \bx^* \Vert]$ for $t\geq 0$. By Lemma \ref{lem: one step} we can write,
	\[ \psi^{t+1} \leq \psi^t(1-\eta_t \mu)^2 + \frac{\eta_t^2 \sigma^2}{N} + \nu^t, \]
	where $\nu^t\geq 0$ and $\nu^t = o(1/t^2)$. 
	It follows,
	\begin{align} \label{eq: ak}
	\psi^k \leq \underbrace{\psi^0 \prod_{t=0}^{k-1} (1-\eta_t \mu)^2}_{S_1} + \underbrace{\sum_{t=0}^{k-1}  \frac{\eta_t^2 \sigma^2}{N} \prod_{l=t+1}^{k-1} (1-\eta_l\mu)^2}_{S_2} + \underbrace{\sum_{t=0}^{k-1} \nu^t \prod_{l=t+1}^{k-1} (1-\eta_l\mu)^2}_{S_3}.
	\qquad \forall k\geq t_0.
	\end{align}
	Next, we will bound each of the terms $S_1$, $S_2$, and $S_3$. Before that, we note that for $t\geq t_0$, 
	\[ 1 - \eta_t \mu = 1 - \frac{2t}{(t+1)^2} \leq 1 - \frac{2}{t}.\]
	Therefore, for $t_2 > t_1 \geq t_0$ we have,
	\begin{multline}
	\prod_{l=t_1}^{t_2 - 1} (1 - \eta_l \mu) \leq  \prod_{l=t_1}^{t_2 - 1} (1 - \frac{2}{l}) = \exp\left( \sum_{l=t_1}^{t_2 - 1} \log(1 - \frac{2}{l})\right) \\
	\leq \exp \left( \sum_{l=t_1}^{t_2 - 1} \frac{-2}{l} \right) \leq \exp \left( 2\log(t_1) - 2\log(t_2) \right) = \left( \frac{t_1}{t_2}\right)^2.
	\end{multline}
	Now we have the tools we need to bound $S_1$, $S_2$, and $S_3$. we have,
	\begin{align}\label{eq: S1}
	S_1 = \Vert \bbx^0 - \bx^* \Vert^2 \prod_{t=0}^{t_0 - 1} (1-\eta_t \mu)^2 \prod_{t=t_0}^{k-1} (1-\eta_t \mu)^2 \leq (1 - \frac{\mu}{L})^{2t_0}\left( \frac{t_0}{k} \right)^4 \Vert \bbx^0 - \bx^* \Vert^2.
	\end{align}
	\begin{align}\label{eq: S2}
	S_2 &= \sum_{t=0}^{t_0-1}  \frac{\eta_t^2 \sigma^2}{N} \prod_{l=t+1}^{t_0-1} (1-\eta_l\mu)^2 \prod_{l=t_0}^{k-1} (1-\eta_l\mu)^2
	+ \sum_{t=t_0}^{k-1}  \frac{\eta_t^2 \sigma^2}{N} \prod_{l=t+1}^{k-1} (1-\eta_l\mu)^2 \nonumber \\ 
	& \leq \frac{\sigma^2}{N}\left[\sum_{t=0}^{t_0-1} \frac{1}{L^2}(1 - \frac{\mu}{L})^{2(t_0 - 1 - t)} \left( \frac{t_0}{k} \right)^4 + 
	\sum_{t=t_0}^{k-1} \left(\frac{2t}{\mu (t+1)^2} \right)^2 \left( \frac{t+1}{k} \right)^4 \right] \nonumber \\
	& = \frac{\sigma^2}{N}\left[\frac{t_0^4}{L^2 k^4}\sum_{t=0}^{t_0-1}(1 - \frac{\mu}{L})^{2t} + 
	\frac{4}{\mu^2 k^4}\sum_{t=t_0}^{k-1} t^2 \right] \nonumber \\
	& \leq \frac{\sigma^2}{N}\left[\frac{t_0^4}{L^2 k^4}\sum_{t=0}^{\infty}(1 - \frac{\mu}{L})^{2t} + 
	\frac{4}{\mu^2 k^4}\sum_{t=1}^{k-1} t^2 \right] \nonumber \\
	& = \frac{\sigma^2}{N}\left[\frac{t_0^4}{L^2 k^4(1 - (1-\frac{\mu}{L})^2)} + 
	\frac{2k(k-1)(2k-1)}{3\mu^2 k^4} \right] \nonumber \\
	& \leq \frac{\sigma^2}{N}\left[\frac{t_0^4}{L\mu k^4} + 
	\frac{4}{3 \mu^2 k} \right] = \frac{4\sigma^2}{3 N\mu^2 k}\left[1 + \frac{3 \mu t_0^4 }{4Lk^3} \right].
	\end{align}
	Next, we show $S_3 = o(1/k)$. 
	Since $\nu^t = o(1/t^2)$, without loss of generality, we can assume there exists $B_1,B_2>0$ such that for any $\epsilon>0$, there exists $k_1\geq t_0$ such that,
	\begin{align*}
	\nu^t \leq \begin{cases}
	\frac{B_1}{(t+1)^2}, \qquad & t \geq 0, \\
	\frac{\epsilon B_2}{(t+1)^2}, \qquad & t \geq k_1.
	\end{cases}
	\end{align*}
	It follows,
	\begin{align*}
	S_3	&\leq \sum_{t=0}^{t_0 - 1} (1 - \frac{\mu}{L})^{2(t_0 - 1 - t)} \left( \frac{t_0}{k}\right)^4 \frac{B_1}{(t+1)^2}  +	\sum_{t=t_0}^{k_1 - 1} \left( \frac{t+1}{k} \right)^4 \frac{B_1}{(t+1)^2}
	+ \sum_{t=k_1}^{k - 1} \left( \frac{t+1}{k} \right)^4 \frac{\epsilon B_2}{(t+1)^2} \\
	& \leq \sum_{t=0}^{\infty } \left( \frac{t_0}{k}\right)^4 \frac{B_1}{(t+1)^2} + 
	\sum_{t=t_0}^{k_1 - 1} \frac{B_1}{k^2} 
	+  \sum_{t=k_1}^{k - 1}  \frac{\epsilon B_2}{k^2} \\
	& \leq \frac{2t_0^4 B_1}{k^4} + \frac{B_1 (k_1 - t_0)}{k^2} + \frac{\epsilon B_2 (k - k_1)}{k^2} \\
	& \leq \frac{2\epsilon t_0 B_1}{k} + \frac{\epsilon B_1}{k} + \frac{\epsilon B_2}{k} = \frac{\epsilon( B_1(2t_0+1) + B_2)}{k}, \qquad \text{for } k\geq \left\lceil\frac{k_1}{\epsilon}\right\rceil.
	\end{align*}
	Thus, 
	\begin{align}\label{eq: S3}
	S_3 = o\left( \frac{1}{k} \right).
	\end{align}
	Plugging \eqref{eq: S1}-\eqref{eq: S3} in \eqref{eq: ak} results,
	\begin{align*}
	\E[\Vert \bbx^k - \bx^* \Vert^2] &\leq \frac{  (1 - \frac{\mu}{L})^{2t_0} t_0^4}{k^4}\Vert \bbx^0 - \bx^* \Vert^2 + \frac{4\sigma^2}{3 N\mu^2 k} + \frac{\sigma^2 t_0^4}{N L \mu k^4} + o\left( \frac{1}{k} \right) \\
	& = \frac{4\sigma^2}{3 N\mu^2 k} + o\left( \frac{1}{k} \right).
	\end{align*}
\end{proof}

\end{document}

%% file: mymacros.tex
\def\F{\mathcal{F}}
\def\E{\mathbb{E}}
\def\R{\mathbb{R}}

\def\O{\mathcal{O}}

\def\V{{\mathcal V}}

\def\I{{\mathcal I}}

\def\1{{\bf 1}}

\def\e{\epsilon}

\def\bA{{\mathbf A}}

\def\bbx{\bar{\mathbf x}}
\def\bbw{\bar{\mathbf w}}
\def\bbg{\bar{\mathbf g}}

\def\bx{{\mathbf x}}
\def\bw{{\mathbf w}}
\def\bz{{\mathbf z}}
\def\bg{{\mathbf g}}
\def\bu{{\mathbf u}}

\def\bv{{\mathbf v}}
\def\by{{\mathbf y}}

\def\hbg{\hat{\mathbf g}}
\def\tbg{\tilde{\mathbf g}}

\newtheorem{theorem}{Theorem}

\newtheorem{lemma}{Lemma}
\newtheorem{corollary}{Corrollary}
\newtheorem{assumption}{Assumption}